\documentclass[runningheads]{jgaa-art}
\usepackage[utf8]{inputenc}
\usepackage[squaren]{SIunits}
\usepackage{bm}
\usepackage{xspace}
\usepackage{enumitem}
\usepackage{pifont}
\PassOptionsToPackage{dvipsnames}{xcolor}
\DeclareGraphicsExtensions{.pdf,.png,.eps,.jpg,.jpeg}
\graphicspath{{figs/}}
\usepackage{subcaption}
\usepackage{pgfplots}
\definecolor{plot-uni}{HTML}{4daf4a}
\definecolor{plot-coast}{HTML}{377eb8}
\definecolor{plot-cluster}{HTML}{e41a1c}

\usepackage{mathtools}
\usepackage{amsmath,amssymb}
\DeclarePairedDelimiter\set{\{}{\}}
\DeclarePairedDelimiter\abs{\lvert}{\rvert}
\DeclarePairedDelimiter\croc{\langle}{\rangle}

\DeclareMathOperator{\x}{\text{x}}
\DeclareMathOperator{\y}{\text{y}}
\def\Oh{\ensuremath{\mathcal{O}}}
\def\sumDist{Distance\xspace}
\def\xOffset{XOffset\xspace}
\def\xHop{IndexOffset\xspace}
\def\sleader{\texttt{s}-leader\xspace}
\def\poleader{\texttt{po}-leader\xspace}
\def\sleaders{\texttt{s}-leaders\xspace}
\def\poleaders{\texttt{po}-leaders\xspace}
\def\doleaders{\texttt{do}-leaders\xspace}

\def\spoleaders{\texttt{s}/\texttt{po}-leaders\xspace}

\def\sandpoleaders{\texttt{s}- and \texttt{po}-leaders\xspace}

\def\sorpoleaders{\texttt{s}- or \texttt{po}-leaders\xspace}
\def\sarea{\texttt{s}-area\xspace}
\def\poarea{\texttt{po}-area\xspace}
\def\cT{\ensuremath{\mathcal{T}}}
\newcommand{\lca}[1]{\ensuremath{\text{lca}(#1)}}
\newcommand{\cro}[1]{\ensuremath{\text{cr}(#1)}}

\usepackage{url}
\usepackage{cite}
\newtheorem{proposition}[theorem]{Proposition}
\PassOptionsToPackage{bookmarksopen=true,
			bookmarksopenlevel=2,}{hyperref}
\usepackage[capitalise,nameinlink]{cleveref}
\crefname{observation}{Observation}{Observations}
\crefname{proposition}{Proposition}{Propositions}

\begin{document}
%% --------------------------------------------------------------------
%   Paper details
%% --------------------------------------------------------------------
\doi{10.7155/jgaa.v29i1.2975}
\Issue{29}{1}{29}{61}{2025} 
\title{Visualizing Geophylogenies -- Internal and External Labeling with Phylogenetic Tree Constraints}
\Ack{Jonathan Klawitter was supported by the Beyond Prediction Data Science Research Programme (MBIE grant UOAX1932). 
Thomas~C.~van~Dijk was partially supported by the DFG grant Di$\,$2161/2-1.
A preliminary version of this paper appeared in the
proceedings of the 12th International Conference on Geographic Information Science (GIScience 2023)~\cite{giscience}.
Implementations of the algorithms and the experiments are
available online at \href{https://www.github.com/joklawitter/geophylo}{github.com/joklawitter/geophylo}.}
\authorOrcid[first]{Jonathan~Klawitter}{jo.klawitter@gmail.com}{0000-0001-8917-5269}
\authorOrcid[second]{Felix~Klesen}{}{0000-0003-1136-5673}
\author[third]{Joris~Y.~Scholl}{}
\authorBOrcid[four]{Thomas~C.~van~Dijk}{}{0000-0001-6553-7317}
\author[second]{Alexander~Zaft}{}
\affiliation[first]{University of Auckland, Aotearoa New Zealand}
\affiliation[second]{Universität Würzburg, Germany}
\affiliation[third]{Ruhr-Universität Bochum, Germany}
\affiliation[four]{Eindhoven University of Technology, Netherlands}
\HeadingAuthor{J.\ Klawitter et al.}
\HeadingTitle{Visualizing Geophylogenies}
%% --------------------------------------------------------------------

%% --------------------------------------------------------------------
%   Paper history
%% --------------------------------------------------------------------
\submitted{Sept. 2024}%
\reviewed{December 2024}%
\revised{January 2025}%
\accepted{March 2025}%
\final{March 2025}%
\published{March 2025}%
\type{Regular paper}%
\editor{Sabine Cornelsen}%
%% --------------------------------------------------------------------

\maketitle        

\pdfbookmark[1]{Abstract}{Abstract} 
\begin{abstract}
A \emph{geophylogeny} is a phylogenetic tree (or dendrogram) where each leaf (e.g.\ biological taxon) has an associated geographic location (site).
To clearly visualize a geophylogeny, the tree is typically represented as a crossing-free drawing next to a map.
The correspondence between the taxa and the sites is either shown with matching labels on the map (internal labeling) 
or with \emph{leaders} that connect each site to the corresponding leaf of the tree (external labeling).
In both cases, a good order of the leaves is paramount for understanding the association between sites and taxa.
We define several quality measures for internal labeling and give an efficient algorithm for optimizing them.
In contrast, minimizing the number of leader crossings in an external labeling is NP-hard.
On the positive side, we show that crossing-free instances can be solved in polynomial time
and give a fixed-parameter tractable (FPT) algorithm.
Furthermore, optimal solutions can be found in a matter of seconds on realistic instances using integer linear programming.
Finally, we provide several efficient heuristic algorithms and experimentally show them to be near optimal on real-world and synthetic instances.
\end{abstract}

\section{Introduction} % -----------------------
\label{sec:intro}
A \emph{phylogeny} describes the evolutionary history and relationships of a set of taxa
such as species, populations, or individual organisms~\cite{Ste16}.
It is one of the main tasks in phylogenetics to infer a phylogeny for some given data and a particular model.
Most often, a phylogeny is modeled and visualized with a \emph{rooted binary phylogenetic tree}~$T$,
that is, a rooted binary tree $T$ where the leaves are bijectively labeled with a set of $n$ taxa.
For example, the phylogenetic tree in \cref{fig:baseExample:tree} shows 
the evolutionary species tree of the five present-day kiwi (\textit{Apteryx}) species.
The term \emph{dendrogram} is used synonoumsly with phylogenetic tree, where the tree represents a hierarchical clustering. 
These trees are conventionally drawn with all edges directed downwards to the leaves and without crossings (\emph{downward planar}).
There exist several other models for phylogenies such as the more general phylogenetic networks,
which can additionally model reticulation events such as horizontal gene transfer and hybridization~\cite{HRS10},
and unrooted phylogenetic trees, which only model the relatedness of the taxa~\cite{Ste16}.
Here we only consider rooted binary phylogenetic trees
and refer to them simply as phylogenetic trees.

In the field of phylogeography, geographic data is used in addition to the genetic data to improve the inference of the phylogeny.
We may thus have spatial data associated with each taxon of a phylogenetic tree
such as the distribution range of each species or the sampling site of each voucher specimen used in a phylogenetic analysis.
For example, \cref{fig:baseExample:map} shows the distributions of the kiwi species from \cref{fig:baseExample:tree}.
Similarly, dendrograms might arise from hierarchical clustering on locations, e.g.\ cities~\cite{exampleCities} or regions~\cite{exampleCoast}.
We speak of a \emph{geophylogeny} (or \emph{phylogeographic tree})
if we have a phylogenetic tree $T$, a map $R$, and a set $P$ of features in $R$
that contains one feature per taxon of $T$; see \cref{fig:baseExample:geophylo} for a geophylogeny of the kiwi species. 
In this paper, we focus on the case where each element $x$ of $P$ is a point, called a \emph{site}, in $R$,
and only briefly discuss the cases where $x$ is a region, or a set of points or regions.

\begin{figure}[t]
  \centering
    \begin{subfigure}[t]{0.20 \linewidth}
		\centering
 		\includegraphics[page=1]{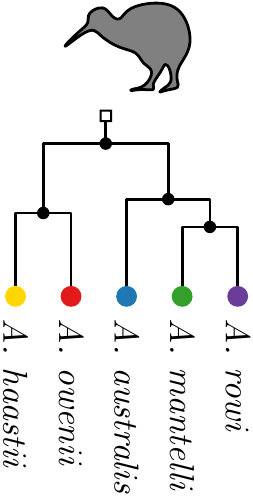}
		\caption{Rooted binary phylogenetic tree.}
		\label{fig:baseExample:tree}
	\end{subfigure}
	\hfill
	\begin{subfigure}[t]{0.30 \linewidth}
		\centering
  		\includegraphics[width=0.908\linewidth]{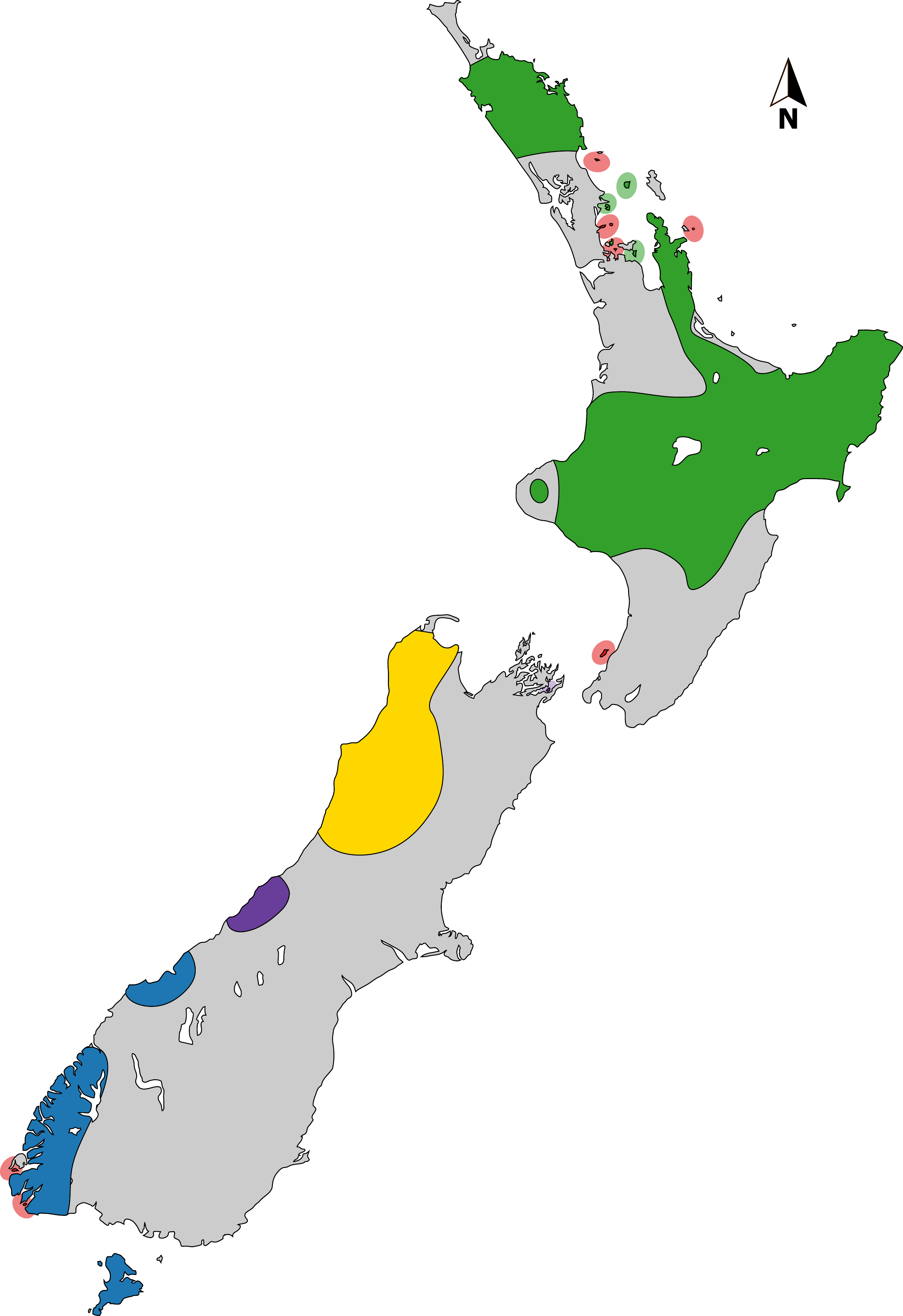}
		\caption{Map with the distribution of each species shown.}
		\label{fig:baseExample:map}
	\end{subfigure}
	\hfill
	\begin{subfigure}[t]{0.40 \linewidth}
		\centering
 		\includegraphics[page=3]{apteryxExample}
 		\hspace{-0.9cm}
 		\includegraphics[height=6cm]{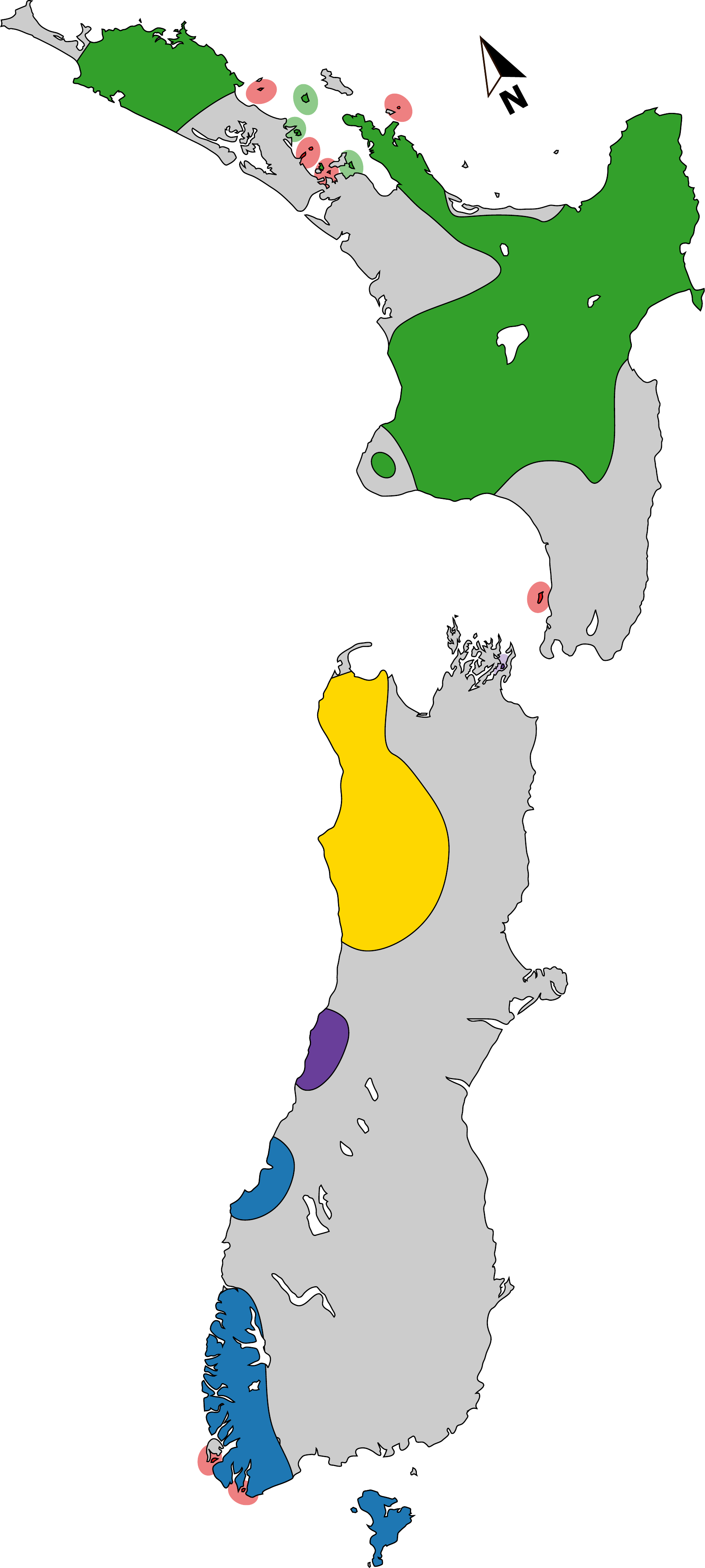}
		\caption{Visualization of the geophylogeny with internal labeling.}
		\label{fig:baseExample:geophylo}
	\end{subfigure}
  \caption{To visualize this geophylogeny of the five present-day 
  kiwi species (Tokoeka/South Island Brown Kiwi -- \textit{Apteryx australis}, Rowi/Okarito Brown Kiwi -- \textit{A. rowi}, 
  North Island Brown Kiwi -- \textit{A. mantelli}, Great Spotted Kiwi -- \textit{A. haastii}, Little Spotted Kiwi -- \textit{A. owenii}), 
  we combine the phylogenetic tree (a) together with the distribution map (b) into a single figure (c). 
  To this end, we may pick a rotation of the map and a placement of the tree
  as well as a leaf order that facilities easy association 
  based on the colors between the leaves and the features on the map. 
  (Phylogeny and map inspired by Weir et~al.\ \cite{exampleKiwi}.)} 
  \label{fig:baseExample}
\end{figure}

\paragraph{Visualizing Geophylogenies.}
When visualizing a geophylogeny, we may want to display its tree and its map together
in order to show the connections (or the non-connections) between the leaves and the sites.
For example, we may want to show that the taxa of a certain subtree
are confined to a particular region of the map or that they are widely scattered.
In the literature, we mainly find three types of drawings of geophylogenies
that fall into two composition categories~\cite{JE12,HSS15}.
In a \emph{side-by-side (juxtaposition)} drawing, 
the tree is drawn planar directly next to the map.
To show the correspondences between the taxa and their sites, 
the sites on the maps are either labeled or color coded (as in \cref{fig:examples:internal} and \cref{fig:baseExample:geophylo}, respectively),
or the sites are connected with \emph{leaders} to the leaves of the tree (as in \cref{fig:examples:external}).
We call this \emph{internal labeling} and \emph{external labeling}, respectively.
There also exist \emph{overlay (superimposition)} illustrations 
where the phylogenetic tree is drawn onto the map in 2D or 3D
with the leaves positioned at the sites~\cite{Xia19,KL08,Rev12}; see \cref{fig:overlay}.
While the association between the leaves and the sites is obvious in overlay illustrations, 
Page~\cite{Pag15} points out that the tree and, in particular, the tree heights might be hard to interpret.

\begin{figure}[t]
	\centering
	\begin{subfigure}[t]{.52 \linewidth}
		\centering
 		\includegraphics[width=1.0\textwidth]{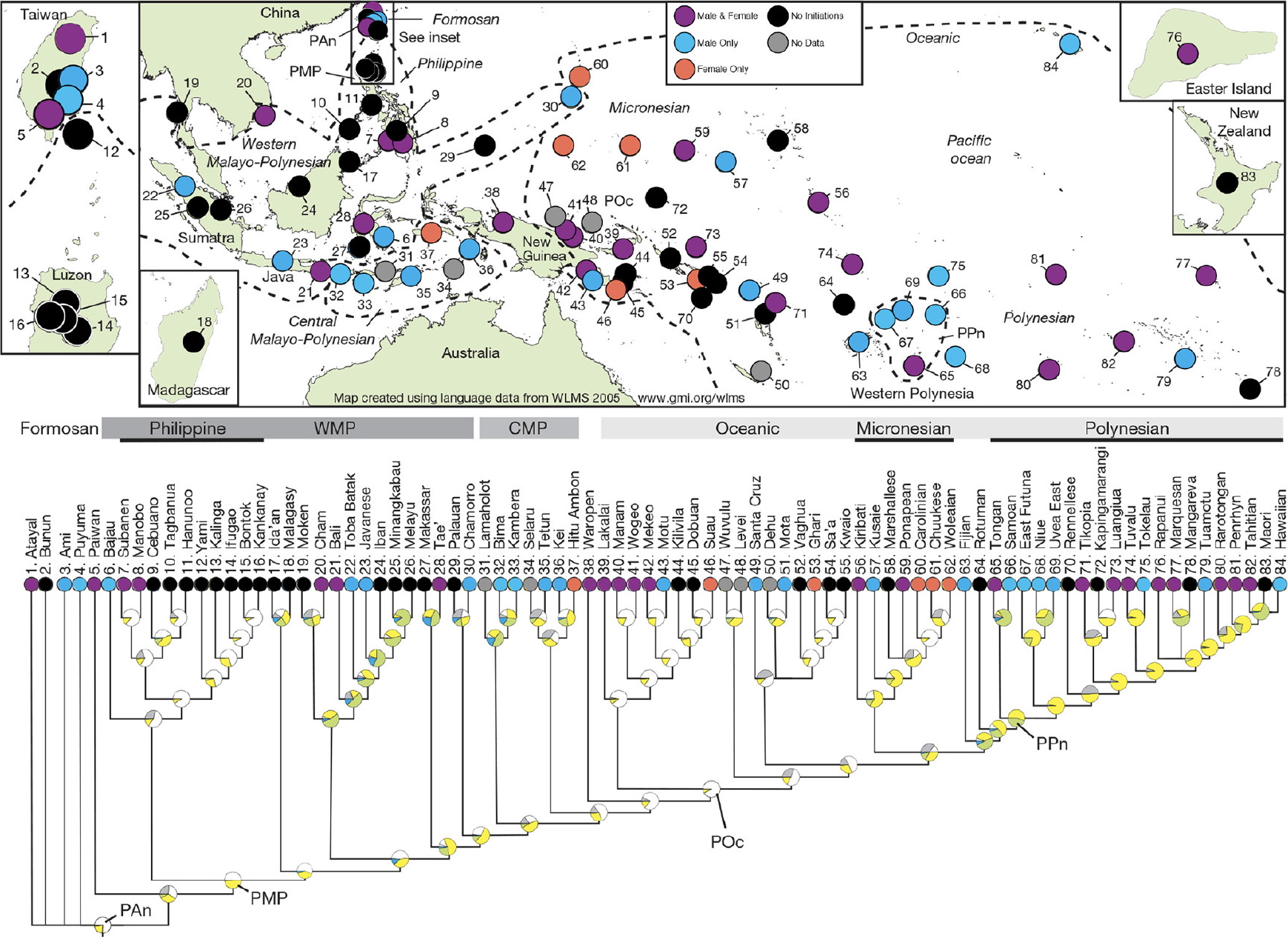}
		\caption{Internal labeling with labels and colors~\cite{internalExample}.}
		\label{fig:examples:internal}
	\end{subfigure}
	\hfill
	\begin{subfigure}[t]{.44 \linewidth}
		\centering
 		\includegraphics[width=1.0\textwidth]{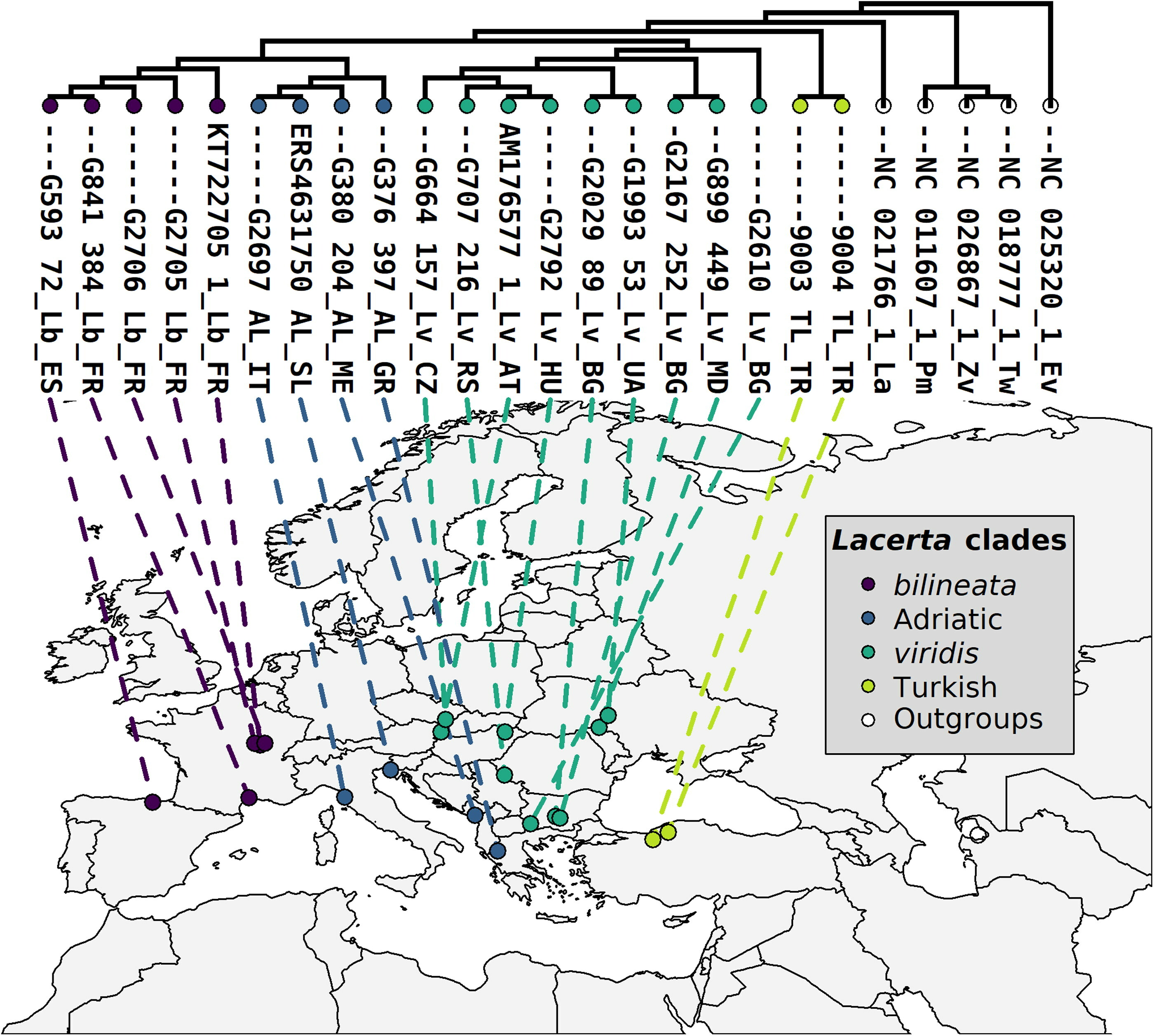}
		\caption{External labeling with \sleaders~~\cite{externalExample}.}		
		\label{fig:examples:external}
	\end{subfigure}
  \caption{Side-by-side drawings of geophylogenies from the literature.} 
  \label{fig:example}
\end{figure}

\begin{figure}[t]
	\centering
	\begin{subfigure}[t]{.48 \linewidth}
		\centering
 		\includegraphics[width=1.0\textwidth]{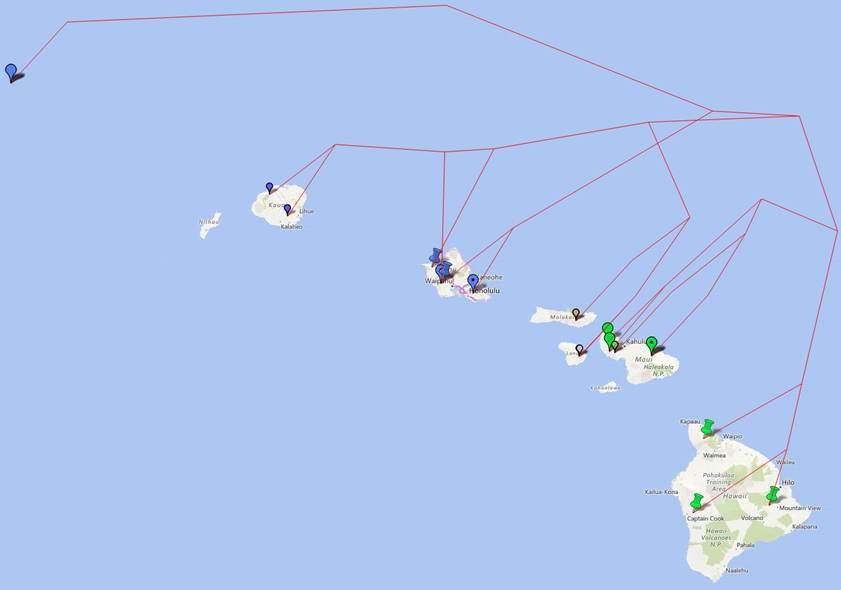}
		\caption{2D overlay drawing of a geophylogeny by Xia~\cite{Xia19}.}
		\label{fig:overlay:two}
	\end{subfigure}
	\hfill
	\begin{subfigure}[t]{.48 \linewidth}
		\centering
 		\includegraphics[width=1.0\textwidth]{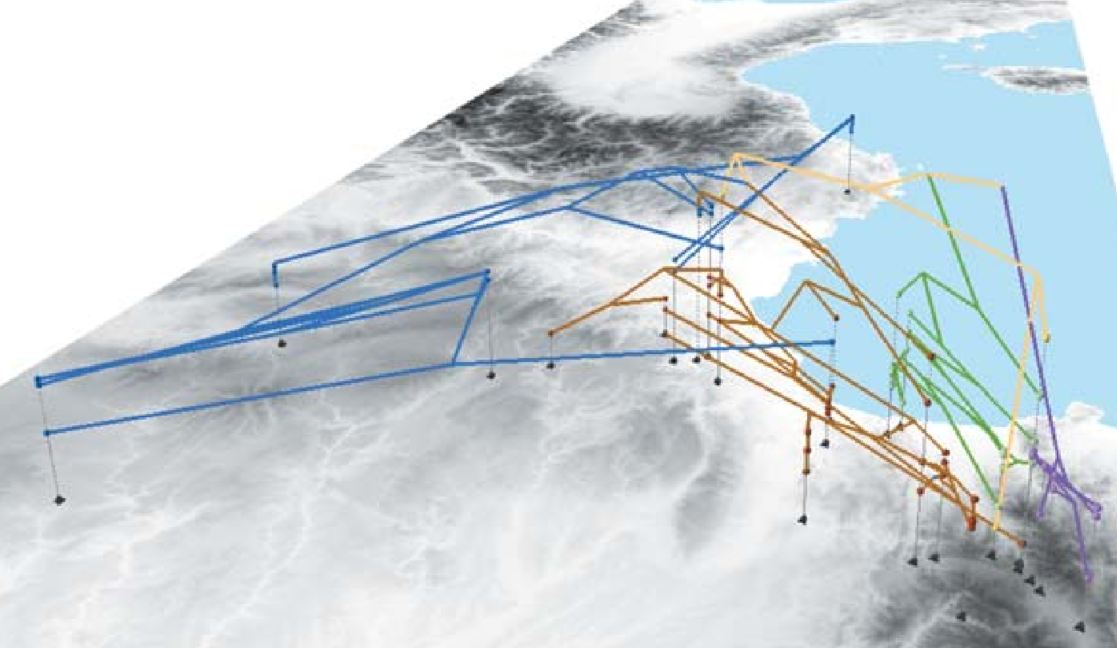}
		\caption{3D overlay drawing of a geophylogeny by Kidd and Liu~\cite{KL08}.}		
		\label{fig:overlay:three}
	\end{subfigure}
  \caption{Overlay drawings of geophylogenies from the literature.} 
  \label{fig:overlay}
\end{figure}

Drawing a geophylogeny involves various subtasks, such as choosing an orientation for the map,
a position for the tree, and the placement of the labels.
Several existing tools support drawing geophylogenies~\cite{PPC+09,PMZ+13,Rev12,Pag15,CPB15},
but we suspect that in practice many drawings are made ``by hand''.
The tools \texttt{GenGIS} by Parks et~al.\ \cite{PPC+09,PMZ+13}, a tool by Page~\cite{Pag15}, % (https://beikolab.cs.dal.ca/gengis/Main\_Page),
and the R-package \texttt{phytools} by Revell~\cite{Rev12} % (http://faculty.umb.edu/liam.revell/, https://github.com/liamrevell/phytools)
can generate side-by-side drawings with external labeling.
The former two try to minimize leader crossings by testing random leaf orders 
and by rotating the phylogenetic tree around the map;
Revell uses a greedy algorithm to minimize leader crossings. 		
The R package \texttt{phylogeo} by Charlop-Powers and Brady~\cite{CPB15} uses internal labeling via colors.
Unfortunately, none of the articles describing these tools formally define a quality measure being optimized
or study the underlying combinatorial optimization problem from an algorithmic perspective.
In this paper, we introduce a simple combinatorial definition 
for side-by-side drawings of geophylogenies and propose several quality measures (\cref{sec:prelim}).

\paragraph{Labeling Geophylogenies.}
The problem of finding optimal drawings of geophylogenies can be considered a special case of map labeling.
In this area, the term \emph{labeling} refers to the process of annotating \emph{features} 
such as points (sites), lines, or regions in maps, diagrams, and technical drawings with labels~\cite{BNN22}.
This facilitates that users understand what they see.
As with geophylogenies, \emph{internal labeling}
places the labels inside or in the direct vicinity of a feature;
\emph{external labeling} places the labels in the margin next to the map
and a label is then connected to the corresponding feature with a \emph{leader}.
An \emph{\sleader} is drawn using a single (straight) line segment as in~\cref{fig:examples:external,fig:geophylogeny:sleader}.
Alternatively, a \emph{\poleader} (for: parallel, orthogonal)
consists of a horizontal segment at the site
and a vertical segment at the leaf, assuming the labels are above the drawing; see~\cref{fig:geophylogeny:poleader}.
In the literature, we have only encountered \sleaders in geophylogeny drawings, 
but argue below that \poleaders should be considered as well.
In a user study on external labeling, Barth, Gemsa, Niedermann, and Nöllenburg~\cite{BGNN19} showed
that users performed best for \poleaders and well for \sleaders when asked to associate sites with their labels and vice versa;
on the other hand, \poleaders and ``diagonal, orthogonal'' \doleaders are the aesthetic~preferences.
We thus consider drawings of geophylogenies that use external labeling with \sandpoleaders.

For internal labeling, a common optimization approach is to place the most labels possible 
such that none overlap; see Neyer~\cite{Ney01} for a survey on this topic. 
Existing algorithms can be applied to label the sites in a geophylogeny drawing 
and it is geometrically straight-forward to place the labels for the leaves of~$T$.
However, a map reader must also be aided in associating the sites on the map 
with the leaves at the border based on these labels (and potentially colors).
Consider the drawing in \cref{fig:baseExample:geophylo}, 
which uses color-based internal labeling: the three kiwi species 
\emph{A. australis}, \emph{A. rowi}, and \emph{A. mantelli} occur in this order from South to North.
When using internal labeling, we would thus prefer, if possible, 
to have the three species in this order in the tree as well --
as opposed to their order in \cref{fig:baseExample:tree}.

External labeling styles conventionally forbid crossings of leaders 
as such crossings could be visually confusing (cf.~\cref{fig:examples:external}).
Often the total length of leaders is minimized given this constraint.
See the book by Bekos, Niedermann, and Nöllenburg~\cite{BNN22} on external labeling techniques.
External labeling for geophylogenies is closely related to many-to-one external labeling,
where a label can be connected to multiple features.
In that case one typically seeks a placement that minimizes the number of crossings between leaders,
which is an NP-hard problem~\cite{LKY08}.
The problem remains NP-hard even when leaders can share segments, so-called hyper-leaders~\cite{BCFHKNRS15}.
Even though our drawings of geophylogenies have only a one-to-one correspondence,
the planarity constraint on the drawing of the tree restricts which leaf orders are possible
and it is not always possible to have crossing-free leaders in a geophylogeny.
In order to obtain a drawing with low visual complexity, 
our task is thus to find a leaf order that minimizes the number of leader crossings.

Note that each vertex of a phylogenetic tree induces a group of labels (leaves) that need to appear consecutive along the boundary,
resulting in non-trivial constraints on the order of the labels.
Niedermann, Nöllenburg, and Rutter~\cite{NNR17} introduced grouping constraints as additional drawing conventions for external labeling,
though did not explore them in detail.
Depian, Nöllenburg, Terziadis, and Wallinger~\cite{DNTW24} studied grouping constraints, which may overlap, as well as ordering constraints 
on labels positioned on one or two sides of the drawings for \poleaders.
They focused on labelings with no crossing leaders and showed that finding label positions is generally NP-hard,
but provided polynomial-time algorithms for practically relevant cases.
On the other hand, Gedicke, Arzoumanidis, and Haunert~\cite{GAH23} incorporated disjoint groupings as optimization criteria to boundary labeling on four sides with \sleaders.
In our case, the phylogenetic tree introduces a set of grouping constraints that can only nest but otherwise not overlap,
placing our problem within the recent focus on labeling under additional constraints.

\paragraph{Further Related Work.}
Since there exists a huge variety of different phylogenetic trees and networks,
it is no surprise that a panoply of software to draw phylogenies has been developed~\cite{HB05,Ram12,Wiki}.
Here we want to mention \texttt{DensiTree} by Bouckaert~\cite{Bou10}.
It draws multiple phylogenetic trees on top of each other for easy comparison in so-called cloudograms and,
relevantly to us, has a feature to extend its drawing with a map for geophylogenies. 
Furthermore, the theoretical study of drawings of phylogenies 
is an active research area~\cite{DS04,DH04,BBS05,Hus09,BGJO19,CDMP20,TK20,KS20,KKNW23}.
In many of these graph drawing problems, 
the goal is to find a leaf order such that the drawing becomes optimal in a certain sense.
This is also the case for \emph{tanglegrams}, where two phylogenetic trees (or dendrograms) on the same taxa
are drawn opposite each other (say, one upward and one downward planar).
Pairs of leaves with the same taxon are then connected with straight-line segments
and the goal is to minimize the number of crossings~\cite{BBBNOSW12}.
This problem is NP-hard if the leaf orders of both trees are variable,
but can be solved efficiently when one side is fixed~\cite{FKP10}.
The latter problem is called the \textsc{One-Sided Tanglegram}~problem
and we make use of the efficient algorithm by Fernau et~al.\ \cite{FKP10} later on.

\paragraph{Results and Contribution.}
We formalize several graph visualization problems in the context of drawing geophylogenies.
We propose quality measures for drawings with internal labeling
and show that optimal solutions can be computed in quadratic time (\cref{sec:internal}).
For external labeling (\cref{sec:external}), we prove that although crossing minimization 
of \sandpoleaders is NP-hard in general, 
it is possible to check in polynomial time if a crossing-free drawing exists.
Moreover, we give a fixed-parameter tractable (FPT) algorithm,
where the parameter captures the number of pairs of sites in inconvenient positions, 
and show that there exist instances with practical relevance
that can be solved efficiently by the FPT algorithm.
Furthermore, we introduce an integer linear program (ILP) and several heuristics for crossing minimization.
We evaluate these solutions on synthetic and real-world examples,
and find that the ILP can solve realistic instances optimally in a matter of seconds
and that the heuristics, which run in a fraction of a second, are often (near-)optimal as well (\cref{sec:experiments}).
We close the paper with a discussion and open~problems.

\section{Definitions and Notation} % -----------------------
\label{sec:prelim}
For a phylogenetic tree~$T$,
let~$V(T)$ be its vertex set,~$E(T)$ its edge set,
$L(T)$ its leaves, and $I(T)$ its internal vertices.
We let $n$ denote the number of leaves of $T$, i.e., $n = \abs{L(T)}$.
For an internal vertex~$v$ of~$T$,
let~$T(v)$ be the subtree rooted at~$v$ and~$n(v) = \abs{L(T(v))}$.
The \emph{clade} of~$v$ is~$L(T(V))$, i.e.\ the set of leaves in the subtree rooted at~$v$.
A \emph{cherry} of~$T$ is a subtree of~$T$ on three vertices such that exactly two are leaves of $T$ and the third is their shared parent. 

A \emph{map}~$R$ is an axis-aligned rectangle and a \emph{site} is a point on~$R$. 
A \emph{geophylogeny}~$G$ consists of a phylogenetic tree~$T(G)$, a map~$R(G)$, a set of points~$P(G)$ in~$R(G)$
as well as a 1-to-1 mapping between~$L(T(G))$ and~$P(G)$.
Call the elements of~$L(T(G)) = \set{\ell_1, \ldots, \ell_n}$ and~$P(G) = \set{p_1, \ldots, p_n}$,
so that without loss of generality the mapping is given by the indices,
that is,~$\ell_i \leftrightarrow p_i$, for~$i \in \set{1, \ldots, n}$. 
For further ease of notation, we only write~$T$, $R$, and~$P$ 
instead of~$T(G)$, $R(G)$, and~$P(G)$, respectively, as~$G$ is clear from the context.

We define a \emph{drawing}~$\Gamma$ of~$G$ 
as consisting of drawings of~$R$ with~$P$ and~$T$ in the plane with the following properties; see \cref{fig:geophylogeny}.
We assume that~$T$ is always drawn at a fixed position above~$R$
such that the leaves of~$T$ lie at evenly spaced \emph{positions} on the upper boundary of~$R$;
the position of the leftmost and rightmost leaf may be fixed arbitrarily.
Furthermore, we require that~$T$ is drawn \emph{downward planar},
that is, all edges of~$T$ point downwards from the root towards the leaves,
and no two edges of~$T$ cross. 
(In our examples we draw~$T$ as a ``rectangular cladogram'', but the exact drawing style is irrelevant given downward planarity.) 
The points of~$P$ are marked using crosses in~$R$ 
and the drawing uses either internal labeling as in \cref{fig:geophylogeny:internal} 
or external labeling with \sorpoleaders as in \cref{fig:geophylogeny:sleader,fig:geophylogeny:poleader}.
For drawings with external labeling, we let~$s_i$ denote the leader that connects~$\ell_i$ and~$p_i$.
We consider $p_i$ part of $s_i$. Two leaders $s_i$ and $s_j$ \emph{cross} if their intersection is non-empty, which can also be $p_i$ or $p_j$.
(We ignore the leaf labels as they do not effect the combinatorics:
they can simply be added in a post-processing step where~$T$ can be moved upwards to create the necessary vertical space.
In practice, we further have to pick appropriate sizes for $R$, the spacing, and fonts such that labels do not overlap horizontally.)

Since the tree is drawn without crossings and the sites have fixed locations, 
the only combinatorial freedom in the drawing~$\Gamma$ is the embedding of~$T$, 
i.e.\ which child is to the left and which is to the right.
Furthermore, since we fixed the relative positions of the map and the leaves, 
note that there is also no ``non-combinatorial'' freedom. 
Hence, an embedding of~$T$ corresponds one-to-one with a left-to-right order of~$L(T)$ 
and we call this the \emph{leaf order}~$\pi$ of~$\Gamma$.
For example, if a leaf~$\ell_i$ is at position 4 in~$\Gamma$, then~$\pi(\ell_i) = 4$.
Further, let~$\x(v)$ and~$\y(v)$ denote the x- and y-coordinate, respectively, 
of a site or leaf~$v$ of~$T$ in~$\Gamma$.
In a slight abuse of terminology, we also call it a drawing of a geophylogeny
even when the leaf order has not been fixed yet. 

\begin{figure}[t]
  \centering
    \begin{subfigure}[t]{0.28 \linewidth}
		\centering
		\includegraphics[page=1]{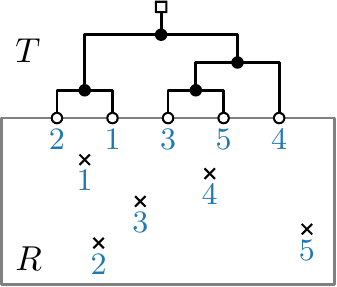}
		\caption{Drawing of $G$ with internal labeling, \ldots}
		\label{fig:geophylogeny:internal}
	\end{subfigure}
	$\quad$
	\begin{subfigure}[t]{0.28 \linewidth}
		\centering
		\includegraphics[page=2]{geophylogenyStyles}
		\caption{\ldots with s-leaders and 1 crossing, \ldots}
		\label{fig:geophylogeny:sleader}
	\end{subfigure}
	$\quad$
	\begin{subfigure}[t]{0.28 \linewidth}
		\centering
		\includegraphics[page=3]{geophylogenyStyles}
		\caption{\ldots and with po-leaders and 0 crossings.}
		\label{fig:geophylogeny:poleader}
	\end{subfigure}
  \caption{In a drawing of a geophylogeny $G$, we place $T$ above $R$
  and use either internal or external labeling to show the mapping between $P$ and~$L(T)$. 
  Figures~(b) and~(c) minimize the number of crossings for their leader type.
  Note the difference in embedding of $T$ and that not all permutations of leaves are possible.}
  \label{fig:geophylogeny}
\end{figure}

\section{Geophylogenies with Internal Labeling} % -----------------------
\label{sec:internal}
In this section, we consider drawings of geophylogenies with internal labeling.
While these drawings trivially have zero crossings -- there are no leaders --
a good order of the leaves is still crucial, since it can help the reader associate between the leaves $L(T)$ and the sites $P$.
It is in general not obvious how to determine which leaf order is best for this purpose;
we propose three quality measures and a general class of measures that subsume them.
Any measure in this class can be efficiently optimized by the algorithm described below.
In practice one can easily try several quality measures and pick whichever suits the particular drawing;
a user study of practical readability could also be~fruitful. 

\pdfbookmark[2]{Quality Measures}{Quality Measures} 
\paragraph{Quality Measures.}
When visually searching for the site $p_i$ corresponding to a leaf $\ell_i$ (or the opposite direction),
it seems beneficial if $\ell_i$ and $p_i$ are close together. 
Our first quality measure, \emph{\sumDist}, 
sums the Euclidean distances of all pairs $(p_i, \ell_i)$; 
see \cref{fig:internal:euclidean}.

Since the tree organizes the leaves from left to right along the top of the map,
and especially if the distance of pairs $\ell_i$ and $p_i$ is dominated by the vertical distance as in \cref{fig:examples:external},
it might be better to consider only the horizontal distances, 
i.e.\ $\sum_{i = 1}^n \abs{\x(p_i) - \x(\ell_i)}$, which we call \emph{\xOffset}; see \cref{fig:internal:x}.
Note that the vertical distance of each leader remains fixed for any leaf order.
Therefore, an optimal solution for \xOffset is equivalent to using the sum of Manhattan distances of all pairs.

Finally, instead of the geometric offset,
\emph{\xHop} considers how much the leaf order permutes the geographic left to right order of the sites.
Assuming without loss of generality that the sites are in general position and indexed from left to right,
we sum how many places each leaf $\ell_i$ is away from leaf position $i$,
i.e.\ $\sum_{i = 1}^n \abs{\pi(\ell_i) - i}$; see \cref{fig:internal:hop}.

These measures have in common that they sum over some ``quality'' of the leaves,
where the quality of a leaf depends only on its own position and that of the sites (but not the other leaves).
Here we call such quality measures \emph{leaf additive}; Benkert, Haverkort, Kroll, and Nöllenburg~\cite{BHKN09} call them \textit{badness functions}
and suggest that a leaf additive quality measure could also take the interference of leaders with the underlying map into account.
Unfortunately not all sensible quality measures are leaf additive (such as for example the number of inversions in $\pi$).

\begin{figure}[tbh]
  \centering
    \begin{subfigure}[t]{0.30 \linewidth}
		\centering
		\includegraphics[page=1]{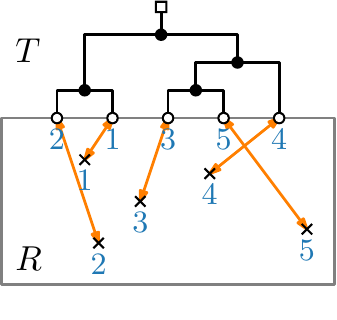}
		\caption{\sumDist~-- sum of Euclidean distances.}
		\label{fig:internal:euclidean}
	\end{subfigure}
	\quad
	\begin{subfigure}[t]{0.30 \linewidth}
		\centering
		\includegraphics[page=2]{geophylogenyInternal}
		\caption{\xOffset~-- sum of horizontal distances.}
		\label{fig:internal:x}
	\end{subfigure}
	\quad
	\begin{subfigure}[t]{0.30 \linewidth}
		\centering
		\includegraphics[page=3]{geophylogenyInternal}
		\caption{\xHop~-- sum of position shifts, here 4.}
		\label{fig:internal:hop}
	\end{subfigure}
  \caption{Orange arrows indicate what the three quality measures for internal labeling consider.}
  \label{fig:internal}
\end{figure}

\pdfbookmark[2]{Algorithm}{Algorithm for Leaf-Additive Quality Measures}
\paragraph{Algorithm for Leaf-Additive Quality Measures.}
Let $f \colon L(T) \times \set{1, \ldots, n} \to \mathbb{R}$ be a quality measure for placing one particular leaf at a particular position; 
the location of the sites is constant for a given instance, so we do not consider it an argument of $f$.
This uniquely defines a leaf additive quality measure on drawings by summing over the leaves;
assume without loss of generality that we want to minimize this sum.

Now we naturally lift $f$ to inner vertices of $T$ by taking the sum over leaves in the subtree rooted at that vertex -- in the best embedding of that subtree.
More concretely, note that any drawing places the leaves of any subtree at consecutive positions and they take up a fixed width regardless of the embedding.
For an inner vertex~$v$, assume that the leftmost leaf of the subtree~$T(v)$ is placed at position~$i$.
Note that because $T(v)$ requires $n(v)$ positions, we have that $i \in \set{1, \ldots, n - n(v) + 1}$.
Let $F(v, i)$ be the minimum sum of the quality $f$ of the leaves of $T(v)$, taken over all embeddings of~$T(v)$.
For $i > n - n(v) + 1$, we set $F(v, i) = \infty$.
Then, by definition, the optimal value for the entire instance is~$F(\rho,1)$, where $\rho$ is the root of~$T$.

\begin{theorem} \label{clm:internal}
Let $G$ be a geophylogeny with $n$ taxa and let $f$ be a leaf additive quality measure.
A drawing $\Gamma$ with internal labeling of $G$ that minimizes (or maximizes) $f$
can be computed in $\Oh(n^2)$~time and $\Oh(n^2)$~space.
\end{theorem}
\begin{proof}
For an inner vertex $v$ with children $x$ and $y$, we observe the following equality, 
since the embedding has only two ways of ordering the children and those subtrees are then independent; 
see also \cref{fig:internal:algo}:
\begin{equation}\label{eq:dp}
	F(v, i) = \min\set{ \quad F(x, i) + F(y, i + n(x)),\quad F(y, i) + F(x, i + n(y)) \quad}
\end{equation}
Using dynamic programming on $F$, for example in postorder over $T$, 
allows us to calculate~$F(\rho,1)$ in $\Oh(n^2)$ time and space, 
since there are $2n$ vertices, $n$ possible leaf positions, 
and \cref{eq:dp} can be evaluated in constant time by precomputing all $n(v)$.
The space requirement is thus also in $\Oh(n^2)$.
As it is typical, the optimal embedding of $T$ can be traced back
through the dynamic programming table in the same running time.
\end{proof}

\begin{figure}[tb]
  \centering
      \begin{subfigure}[t]{0.48 \linewidth}
		\centering
		\includegraphics[page=1]{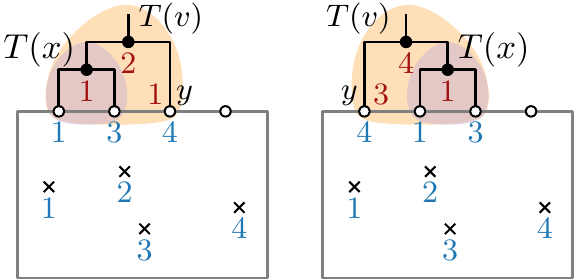}
		\caption{$F(v, 1) = \min\set{2, 4} = 2$}
		\label{fig:internal:algo:left}
	\end{subfigure}
	\hfill
	\begin{subfigure}[t]{0.48 \linewidth}
		\centering
		\includegraphics[page=2]{geophylogenyAlgo}
		\caption{$F(v, 2) = \min\set{1, 5} = 1$}
		\label{fig:internal:algo:right}
	\end{subfigure}
  \caption{Computing \xHop for the subtree $T(v)$ at positions $1$ and $2$.}
  \label{fig:internal:algo}
\end{figure}

\pdfbookmark[2]{Adaptability}{Adaptability} 
\paragraph{Adaptability.}
Note that we can still define leaf additive quality measures
when $P$ contains regions (rather than just points) as in \cref{fig:baseExample}.
For example, instead of considering the distance between~$\ell_i$ and $p_i$,
we could consider the smallest distance between $\ell_i$ and any point in the region~$p_i$.
Similarly, if each element of $P$ is a set of sites, 
we could use the average or median distance to the sites corresponding to $\ell_i$.
For such a leaf additive quality measure $f$,
our algorithm finds an optimal leaf order in $\Oh(n^2 x)$ time 
where~$x$ is a bound on the time needed to compute $f(\ell_i, j)$ over all $i, j \in \set{1, \ldots, n}$.

\pdfbookmark[2]{Interactivity}{Interactivity}
\paragraph{Interactivity.}
With the above algorithm, we can restrict leaves and subtrees to be in a certain position or a range of positions, 
simply by marking all other positions as prohibitively expensive in~$F$;
the rotation of an inner vertex can also be fixed by considering only the corresponding term of~\cref{eq:dp}.
This can be used if there is a conventional order for some taxa
or to ensure that an outgroup-taxon (i.e.\ taxon only included to root and calibrate the phylogenetic tree)
is placed at the leftmost or rightmost position.
Furthermore, this enables an interactive editing experience 
where a designer can inspect the initial optimized drawing 
and receive re-optimized versions based on their feedback
-- for example ``put the leaves for the sea lions only where there is water on the edge of the map''.
(This is leaf additive.)

\section{Geophylogenies with External Labeling} % -----------------------
\label{sec:external}
In this section, we consider drawings of geophylogenies that use external labeling.
Recall that for a given geophylogeny $G$, 
we want to find a leaf order $\pi$ such that the number of leader crossings in a drawing~$\Gamma$ of~$G$ is minimized.   
We show the following.
\begin{enumerate}
	\item The problem is NP-hard in general (\cref{sec:np}). 
	\item A crossing-free solution can be found in polynomial time if it exists (\cref{sec:crossingfree}).
	\item Some instances have a geometric structure that allows us to compute optimal solutions in polynomial time (\cref{sec:easy}).
	\item The problem is fixed parameter tractable (FPT) in a parameter based on this structure (\cref{sec:fpt}).
	\item We give an integer linear program (ILP) to solve the problem (\cref{sec:ilp}).
	\item We give several heuristic algorithms for the problem  (\cref{sec:heuristics}).
\end{enumerate}
All results hold analogously for \sandpoleaders;
only the parameter value of the FPT algorithm is different depending on the leader type.

\subsection{NP-Hardness} % - - - - - - - - - - - -
\label{sec:np}
In order to prove that the decision variant of our crossing minimization problem is NP-complete,
we use a reduction from the classic \textsc{Max-Cut} problem,
which is known to be NP-complete~\cite{GJ79}.
In an instance of \textsc{Max-Cut},
we are given a graph $H$ and a positive integer $c$,
and have to decide if there exists a bipartition $(A, B)$ of $V$
such that at least $c$ edges have one endpoint in $A$ and one endpoint in $B$; 
see \cref{fig:np:maxCut}.
The proof of the following theorem is inspired by the construction Bekos~et~al.\ \cite{BCFHKNRS15} 
use to show NP-completeness of crossing-minimal labeling with hyperleaders 
(with a reduction from \textsc{Fixed Linear Crossing Number}).

\begin{figure}[htb]
  \centering
  \includegraphics[page=1]{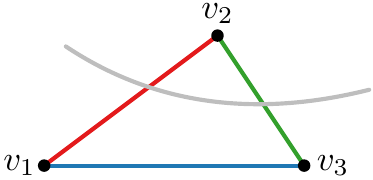}
  \caption{The partition of the triangle $v_1v_2v_3$ cuts the two edges $\set{v_1,v_2}$, $\set{v_2,v_3}$.
  We reduce the corresponding \textsc{Max-Cut} instance to a geophylogeny drawing in~\cref{fig:np:overview}.}
  \label{fig:np:maxCut}
\end{figure}

\begin{theorem} \label{clm:po:np}
Let $G$ be a geophylogeny and $k$ a positive integer.
For both \sandpoleaders,
it is NP-complete to decide whether a drawing $\Gamma$ of $G$ with external labeling and a leaf order $\pi$
that induces at most $k$ leader crossings exists.
\end{theorem}
\begin{proof}
The problem is in NP since, given $G$, $k$, and $\pi$, 
we can check in polynomial time whether this yields at most $k$ crossings.
To prove NP-hardness, we use a reduction from \textsc{Max-Cut} as follows.
The proof works the same for \sandpoleaders; we use \poleaders in the figures.

For an instance $(H, c)$ of \textsc{Max-Cut},
we construct an instance of our leader crossing minimization problem 
by devising a geophylogeny $G$ with phylogenetic tree $T$, points $P$ on a map $R$ 
and a constant $k$; see \cref{fig:np:overview}.
Without loss of generality, we assume that each vertex in $H$ has at least degree 2.
Let $V(H) = \set{v_1, \ldots, v_n}$ and $m = \abs{E(H)}$.
We consider each edge $\set{v_i, v_j}$ with $i < j$ to be directed as $v_iv_j$.
Let $E(H)$ be ordered $e_1, \ldots, e_m$ lexicographically on the indices $i$ and $j$. 
Throughout the following, let $(A, B)$ be some partition of $V(H)$ and 
let $R$ have height $4m + 4 + d$ where we set $d$ appropriately below.

We first describe the broad structure of the reduction 
and then give details on the specific gadgets. 
Each vertex is represented by a \emph{vertex gadget} in $T$.
For each edge $v_iv_j$ in $E(H)$, there is an \emph{edge gadget} that connects sites
on the map to the vertex gadgets with four leaders.
Using \emph{fixing gadgets} to restrict the possible positions 
for vertex gadget's leaves, we enforce that an edge gadget induces 2 crossing 
if $v_i$ and $v_j$ are both in $A$ or both in $B$;
otherwise it will induce 1 crossing.
The number of crossings between leaders of different edge gadgets is in total some constant $k_{\text{fix}}$.
We set $k = k_{\text{fix}} + 2m - c$.
Consequently, if $G$ admits a drawing with at most~$k$ leader crossings, 
then $H$ admits a cut with at least $c$ edges, and vice versa.
\begin{figure}[t]
  \centering
  \includegraphics[page=2]{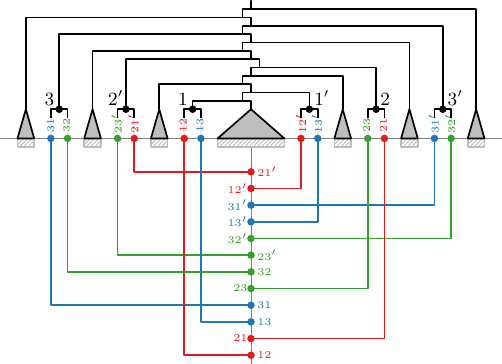}
  \caption{Sketch of the reduction of the graph from \cref{fig:np:maxCut} to a geophylogeny drawing with \poleaders.
  We simplified $v_i$ to $i$; each edge gadget is drawn in the respective color;  
  fixing gadgets are represented by triangles in the tree and hatched rectangles on the map.}
  \label{fig:np:overview} 
\end{figure}

\paragraph{Vertex Gadgets.}
Each vertex~$v_i \in V(H)$ is represented by two subtrees 
rooted at vertices~$i$ and~$i'$ in~$T$
such that from~$i$ going two edges up and one down we reach~$i'$. 
In~$T(i)$ there is a leaf labeled~$ij$ for each edge~$v_iv_j$ or~$v_jv_i$ incident to~$v_i$ in~$H$.
Furthermore,~$T(i)$ has a planar embedding where the leaves can be in increasing (or decreasing)
order based on the order of the corresponding edges in~$E(H)$; see again \cref{fig:np:overview}.
$T(i')$ is built analogously, though we label the leaves with~$ij'$.
In~$T$, the vertex gadgets and fixing gadgets alternate;
more precisely, the subtree of a central fixing gadget
lies inside the subtree of the vertex gadget for~$v_1$,
which in turn lies in the subtree of a fixing gadget, and so on.  
The fixing gadgets ensure that either~$T(i)$ is in the left half of the drawing
and~$T(i')$ in the right half, or vice versa (explained below).
Furthermore, we interpret~$T(i)$ being in the left (right) half as~$v_i$ being in~$A$ (resp.~$B$).

\paragraph{Edge Gadgets.}
For an edge~$e_h = v_iv_j \in E(H)$, we have four sites~$ij$,~$ji$,~$ij'$,~$ji'$ on the central axis of the drawing,
which correspond to the leaves in~$T(i)$,~$T(j)$,~$T(i')$,~$T(j')$ with the same label.
From bottom to top, we place the sites~$ij$ and~$ji$ at heights~$2h-1$ and~$2h$, respectively;
we place the sites~$ij'$ and~$ji'$ at~$4m-2(h-1)-1$ and~$4m-2(h-1)$, respectively; see \cref{fig:np:edgeGadget}.
Hence, in the bottom half the sites are placed in the order of the edges,
while in the top half they are (as pairs) in reverse order.
Note that while the order of the sites~$ij$,~$ji$,~$ij'$,~$ji'$ is fixed,
the order of the leaves~$ij$,~$ji$,~$ij'$,~$ji'$ is not.
Yet there are only four possible orders
corresponding to whether~$v_i$ and~$v_j$ are in~$A$ or~$B$.
Further note that whether the leaders of the edge gadget cross 
is therefore not based on the geometry or the type of the leaders but solely on the leaf order. 
In particular, if~$v_iv_j$ is cut by~$(A,B)$ (as in \cref{fig:np:edgeGadget:cut}), 
then we have the leaf order~$ji'$,~$ij$,~$ij'$,~$ji$
with~$ji'$ and~$ij$ left of the center (up to reversal of the order).
Therefore the leaders~$s_{ij}$ and~$s_{ji'}$ cross while~$s_{ij'}$ and~$s_{ji}$ do not.
Hence, there is exactly one crossing.
On the other hand, if~$v_iv_j$ is not cut by~$(A,B)$ (as in \cref{fig:np:edgeGadget:notcut}),
then we have the leaf order~$ij$,~$ji$,~$ij'$,~$ji'$
with~$ij$ and~$ji$ left of the center (up to reversal of the order).
Hence we have two crossings as both~$s_{ij}$ and~$s_{ji}$ as well as~$s_{ij'}$ and~$s_{ji'}$ cross.

\begin{figure}[h]
	\centering
	\begin{subfigure}[t]{.47 \linewidth}
		\centering
		\includegraphics[page=3]{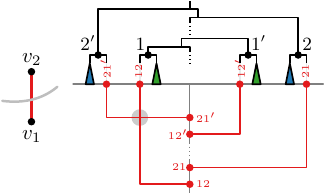}
		\caption{If $v_1v_2$ is in the cut, its edge gadget
		induces exactly 1 crossing.}
		\label{fig:np:edgeGadget:cut}
	\end{subfigure}
	\hfill
	\begin{subfigure}[t]{.47 \linewidth}
 		\centering
		\includegraphics[page=4]{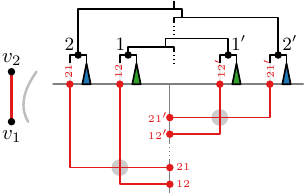}
		\caption{If $v_1v_2$ is not in the cut, its edge gadget
		induces exactly 2 crossing.}
		\label{fig:np:edgeGadget:notcut}
	\end{subfigure}
  \caption{The edge gadget for $v_1v_2$ connects the vertex gadgets for $v_1$ and $v_2$.}
  \label{fig:np:edgeGadget} 
\end{figure}

\paragraph{Edge Pairs.}
Let~$v_iv_j,v_kv_l \in E(H)$. 
We assume an optimal leaf order in each vertex gadget. 
Then careful examination of the overall possible leaf orders (and partitions) shows
that the leaders in the edge gadgets of $v_iv_j$ and $v_kv_l$
induce exactly three crossings if $v_iv_j$ and $v_kv_l$ share a vertex; see again \cref{fig:np:overview}.
If the two edges are disjoint, then the leaders induce exactly four crossings; see \cref{fig:np:edgePairs}.
Note that changing the partition or the order of vertices does not change the number of crossings; 
it only changes which pairs among the eight leaders cross.
We can thus set $k_{\text{fix}}$ as three times the number of adjacent edge pairs 
plus four times the number of disjoint edge pairs.

\begin{figure}[tbh]
  \centering
  \includegraphics[page=5]{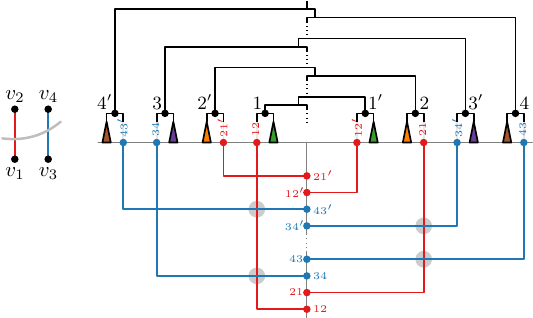}
  \caption{Leaders for two edge gadgets of disjoint edges induce four crossings.}
  \label{fig:np:edgePairs} 
\end{figure}

\paragraph{Fixing Gadgets.}
To ensure that the two subtrees of each vertex gadget
are distributed to the left and to the right, 
we add a fixing gadget in the center and
one after each position allocated to a vertex gadget subtree. 
If both subtrees of a vertex gadget would be placed on the same side of a fixing gadget,
then the fixing gadget would have to be translated and induce too many crossings.
More precisely, each fixing gadget is composed of a series of \emph{fixing units}.
A fixing unit~$F$ consists of a four-leaf tree
with cherries $\set{a, a'}$ and $\set{b, b'}$.
Assuming $F$ is to be centered at position $x$,
we place the sites for $a$ and $a'$ (for $b$ and $b'$) at $x$ 
at height $4m + d + 1$ and 4 (resp. plus 2 and 3), respectively.
Thus if $F$ is centered at $x$, it can be drawn with 0 crossings; see \cref{fig:np:fixing:inplace}.
However, if $F$ is translated by two or more then it induces 2 crossings; see \cref{fig:np:fixing:translated}.
Since each vertex of $H$ has at least degree two,
the two trees $T(i)$ and $T(i')$ of a vertex gadget have at least two leaves each.
Hence, $F$ cannot be translated by just one position.
By using $m-c$ fixing units per fixing gadget, 
it becomes too costly to move even one fixing gadget 
as the instance would immediately have to many crossings.

Finally, we set $d$ such that no leader of an edge gadget can cross a leader of a fixing gadget.
In particular, $d = 4$ is sufficient for \poleaders. 
Note that a $\pm 45^\circ$ slope for \sleaders suffices to avoid crossing a fixing gadget.
Between the leftmost (rightmost) vertex gadget leaf and the center are at most $n$ fixing gadgets, each with $m-c$ fixing units of width 4, and $m - 1$ other vertex gadget leaves.
We can bound this from above with $4m^2$ and thus set $d = 4m^2$ for \sleaders.
\end{proof}

\begin{figure}[tbh]
	\centering
	\begin{subfigure}[t]{.46 \linewidth}
		\centering
		\includegraphics[page=6]{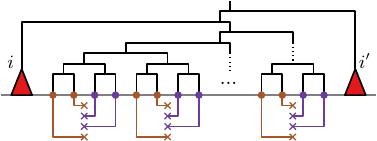}
		\caption{A fixing gadget at its designated place induces zero crossings.}
		\label{fig:np:fixing:inplace}
	\end{subfigure}
	\hfill
	\begin{subfigure}[t]{.46 \linewidth}
		\centering
		\includegraphics[page=7]{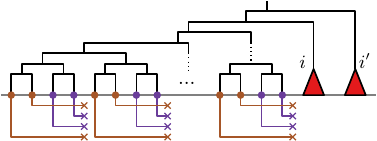}
		\caption{A fixing gadget translated by at least two induces $2$ crossings per fixing unit.}
		\label{fig:np:fixing:translated}
	\end{subfigure}
  \caption{A fixing gadget consists of a series of four-leaf fixing units and is always placed at its designed place
  since it would otherwise cause too many crossings.}
  \label{fig:np:fixing}
\end{figure}

Note that the construction used in the proof of \cref{clm:po:np} does not rely on the sites of the edge gadget or the fixing units being collinear.
As long as the sites of the edge gadgets (a fixing unit) lie ``close enough'' to the center (resp. the center of the fixing unit) and maintain their relative vertical order, they can be in general position.

\subsection{Crossing-Free Instances} % - - - - - - - - - - - -
\label{sec:crossingfree}
We now show how to decide whether a geophylogeny admits a drawing without leader crossings in polynomial time
for both \sandpoleaders.

\begin{proposition} \label{clm:s:zero}
Let $G$ be a geophylogeny on $n$ taxa.
For both \sandpoleaders, we can decide if a drawing $\Gamma$ of $G$ with external labeling
and a leaf order $\pi$ that induces zero leader crossings exists in~$\Oh(n^6)$~time.
\end{proposition}
\begin{proof}
To find a leaf order~$\pi$ for a drawing~$\Gamma$ that induces zero leader crossings, if it exists,
we use a dynamic program similar to the one we used for internal labeling in \cref{clm:internal}.
Let~$i \in \set{1, \ldots, n}$ and let~$v \in V(T)$.
Then we store in~$F(v, i)$ up to~$n(v)$ embeddings of~$T(v)$
for which~$T(v)$ can be placed with its leftmost leaf at position~$i$
such that the leaders to~$T(v)$ are pairwise crossing free.
Note that~$F(v, i)$ always stores exactly one embedding when~$v$ is a leaf.
For an inner vertex~$v$ with children~$x$ and~$y$,
we combine pairs of stored embeddings of~$T(x)$ and~$T(y)$
and test whether they result in a crossing free embedding of~$T(v)$.
For~$\rho$ the root of~$T$, we get a suitable leaf order for each embedding stored in~$F(\rho, 1)$.
However, since combining embeddings of~$T(x)$ and~$T(y)$ can result in $\Oh(n(v)^2)$ many embeddings of $T(v)$,
we have to be more selective.
We now describe when we have to keep multiple embeddings of~$T(v)$, how we select them,
and show that at most~$n(v)$ embeddings for~$T(v)$ at position~$i$ suffice. 
We describe the details first for \sleaders and then for \poleaders.

\paragraph{\texttt{s}-Leaders.}
Suppose that we can combine an embedding of $T(x)$ and an embedding of $T(y)$
where~$T(v)$ is placed with its leftmost leaf at position~$i$
such that the leaders of $T(v)$ pairwise do not cross.
Consider the set~$P(v)$ of sites corresponding to~$L(T(v))$.
In particular, let~$p_k$ have the lowest y-coordinate among the sites in~$P(v)$.
Let~$H(v, i)$ be the convex hull of the sites~$P(v)$ and the leaf positions~$i$ and~$i + n(v) -1$; see \cref{fig:zeroCrossings:s}.
We distinguish three cases:

\begin{enumerate}[leftmargin=*,label=Case \arabic* - ]
  \item \textbf{there is no site of~$\bm{ P \setminus P(v) }$ inside~$\bm{ H(v, i) }$:} 
    Then no leader of a site~$p_o \in P \setminus P(v)$ has to ``leave''~$H(v, i)$.
    A leader that would need to intersect~$H(v, i)$ would cause a crossing with a leader of~$T(v)$
    for any embedding of~$T(v)$.
  	Hence it suffices to store only this one embedding of~$T(v)$ and not consider any further embeddings.

  \item \textbf{there is a site~$\bm{ p_o \in P \setminus P(v) }$ trapped in~$\bm{ H(v, i) }$:}
  	More precisely, let~$H(v, i, p_o)$ be the convex hull of the positions~$i$ and~$i + n(v) -1$ 
	and all sites of~$P(v)$ above~$p_o$.
	We consider~$p_o$ \emph{trapped} if the leader of~$p_o$ cannot 
	reach any position left of~$i$ or right of~$i + n(v) -1$
	without crossing~$H(v, i, p_o)$; see \cref{fig:zeroCrossings:s:blocked}.
	Hence we would definitely get a crossing for this embedding of~$T(v)$ later on
	and thus reject it immediately.
	
  \item \textbf{there is a site~$\bm{ p_o \in P \setminus P(v) }$ but not trapped inside~$\bm{ H(v, i) }$:}
    Suppose that the leader~$s_o$ of~$p_o$ can reach positions~$j_1, \ldots, j_{k_o}$ without intersecting~$H(v, i, p_o)$.
  	Consider the leader~$s_k$ of~$p_k$ for the current embedding of~$T(v)$.
  	Note that~$s_k$ prevents~$s_o$ from reaching either any position to the left of~$i$
  	or to the right of~$i + n(v) -1$; see \cref{fig:zeroCrossings:s:left}. 
  	If this means that $s_k$ cannot reach any position, then we reject the embedding.
  	Otherwise we would want to store this embedding of~$T(v)$ and 
  	an embedding of $T(v)$ where $s_k$ can reach a position on the ``other'' side of $p_o$ (if it exists).
  	However, we have to consider all other sites of~$P \setminus P(v)$ in~$H(v, i)$, which we do as follows. 
\end{enumerate}

\begin{figure}[htb]
  \centering
	\begin{subfigure}[t]{0.32 \linewidth}
		\centering
		\includegraphics[page=1]{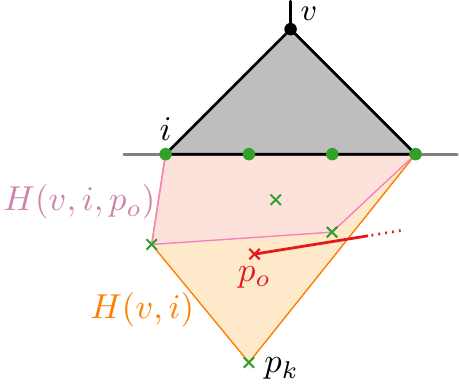}
		\caption{The non-$T(v)$ site~$p_o$ inside~$H(v, i)$
		may connect only to a position far to the right.}
		\label{fig:zeroCrossings:s:blocked}
	\end{subfigure} 
	\hfill
    \begin{subfigure}[t]{0.32 \linewidth}
		\centering
		\includegraphics[page=2]{sZeroCrossing}
		\caption{The leader of the non-$T(v)$ site~$p_o$ inside~$H(v, i)$ cannot
		reach positions to the left of position~$i$.}
		\label{fig:zeroCrossings:s:left}
	\end{subfigure}
	\hfill
	\begin{subfigure}[t]{0.32 \linewidth}
		\centering
		\includegraphics[page=3]{sZeroCrossing}
		\caption{The leader of one non-$T(v)$ site can go the left,
		the two others can go to the right. }
		\label{fig:zeroCrossings:s:partition}
	\end{subfigure}
  \caption{Trying to find a leaf order that induces zero crossings of \sleaders, 
  we store or reject embeddings of~$T(v)$ based on non-$T(v)$ sites (i.e.\ sites in $P \setminus P(v)$) in~$H(v, i)$.}
  \label{fig:zeroCrossings:s}
\end{figure}

There are at most~$n - n(v)$ others sites in~$H(v)$.
If any of them is trapped, we reject the embedding.
Assume otherwise, namely that for the current embedding, all of them can reach a position outside of $T(v)$.
The leader of~$p_k$ then partitions these sites into those 
that can go out to the left and those that can go out to the right.
Hence, among all suitable embeddings of $T(v)$
these sites can be partitioned in at most~$n(v)$ different ways 
(since the leader of $p_k$ can go to only that many positions); 
see \cref{fig:zeroCrossings:s:partition}.
Furthermore, since we have \sleaders, the choice of positions for the other sites of $T(v)$
only influence whether another site is trapped
but not which positions their leaders can reach.
So for each such partition, we need to store only one embedding.
Therefore, before storing a suitable embedding of~$T(v)$, 
we first check whether we already store an embedding where~$\ell_k$ is at the same position. 

We can handle each of the~$\Oh(n(v)^2)$ embeddings of~$T(v)$ in~$\Oh(n^2)$ time each.
With~$n$ positions and~$\Oh(n)$ vertices, we get a running time in~$\Oh(n^6)$.

\paragraph{\texttt{po}-Leaders.}
As with \sleaders, we want to store at most~$\Oh(n(v))$ embeddings of~$T(v)$ for \poleaders.
Let~$H'(v, i)$ be the rectangle that horizontally
spans from positions~$i$ to~$i + n(v) -1$ and vertically from~$p_k$ to the top of~$R$. 
For the current embedding of $T(v)$ and for any site~$p_o \in P \setminus P(v)$ that lies insides~$H'(v, i)$,
we check whether the horizontal segment of the leader $s_o$ of $p_o$ can leave~$H'(v, i)$
without intersecting a vertical segment of a leader of $T(v)$.
If this is not the case for a leader, then we reject the embedding; see \cref{fig:zeroCrossings:po:blocked}.
Otherwise, the leader~$s_k$ of~$p_k$ determines for each $s_o$
whether it can leave~$H'(v, i)$ on the left or on the right side.
Therefore, $s_k$ partitions the sites in~$P \setminus P(v)$ that lie insides~$H'(v, i)$
and we need to store only one suitable embedding for each partition; see \cref{fig:zeroCrossings:po:partition}.
Note that the horizontal segments of the leader $s$ of any site of $P(v)$
that lies outside of~$H'(v, i)$ always spans to at least~$H'(v, i)$.
Therefore whether $s$ intersects with another leader later on outside of $H'(v,i)$
is independent of the embedding of $T(v)$.  
The running time for \poleaders is the same as for \sleaders
and thus also in~$\Oh(n^6)$. 
\end{proof}

\begin{figure}[htb]
  \centering
	\begin{subfigure}[t]{0.44 \linewidth}
		\centering
		\includegraphics[page=1]{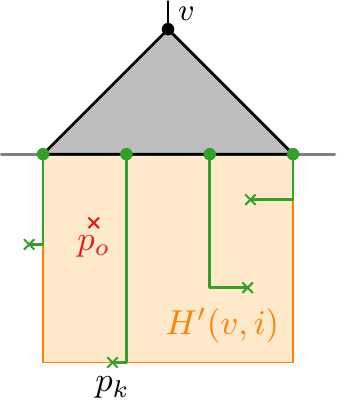}
		\caption{The non-$T(v)$ site~$p_o$ inside $H'(v,i)$
		is trapped by vertical segments of leaders of $T(v)$.} 
		\label{fig:zeroCrossings:po:blocked}
	\end{subfigure}
	$\quad$$\quad$
    \begin{subfigure}[t]{0.44 \linewidth}
		\centering
		\includegraphics[page=2]{poZeroCrossing}
		\caption{The leader of $p_k$ partitions the non-$T(v)$ sites inside $H'(v, i)$ 
		into whether their leader leaves~$H'(v, i)$ to the left or the right.}
		\label{fig:zeroCrossings:po:partition}
	\end{subfigure}
  \caption{Trying to find a leaf order that induces zero crossings of \poleaders, 
  we store or reject embeddings of~$T(v)$ based on other sites in~$H'(v, i)$.}
  \label{fig:zeroCrossings:po}
\end{figure}

\subsection{Efficiently Solvable Instances} % - - - - - - - - - - - -
\label{sec:easy}
We now make some observations about the structure of geophylogeny drawings.
This leads to an $\Oh(n\log n)$-time algorithm for crossing minimization 
on a particular class of ``geometry-free'' instances, 
and forms the basis for our FPT algorithm and ILP.

Consider a drawing~$\Gamma$ of a geophylogeny~$G$ with \sleaders and leaf order~$\pi$.
Let $B$ be the line segment between leaf position 1 (left) and leaf position $n$ (right);
let the \emph{\sarea} of a site $p_i$ be the triangle spanned by $p_i$ and $B$.
Note that the leader~$s_i$ lies within this triangle in any drawing.
Now consider two sites $p_i$ and $p_j$ that lie outside each other's \sarea.
Independently of the embedding of the tree, $s_i$ always passes $p_j$ on the same side: 
see \cref{fig:easy} where, for example, $s_2$ passes left of~$p_4$ in any drawing.
As a result, if $p_i$ lies left of~$p_j$,
then $s_i$ and $s_j$ cross if and only if 
the leaf $\ell_i$ is positioned right of the leaf $\ell_j$,
i.e.\ $\pi(\ell_i) > \pi(\ell_j)$.
The case where $p_i$ is right of $p_j$ is flipped.
We call such a pair $\croc{p_i, p_j}$ \emph{geometry free} 
since purely the \emph{order} of the corresponding leaves suffices to recognize if their leaders cross: 
the precise geometry of the leaf positions is irrelevant.
Note that symmetrically $\croc{p_j, p_i}$ is also geometry free.

\begin{figure}[bth]
  \centering
  	\begin{subfigure}[t]{0.42 \linewidth}
		\centering
		\includegraphics[page=1]{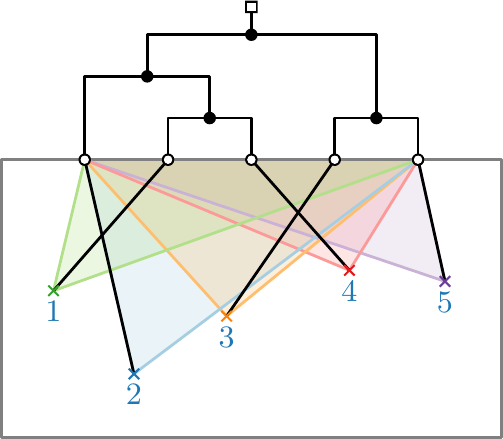}
		\caption{A geometry-free drawing for \sleaders: 
		no site lies inside the \sarea of another~site.}
		\label{fig:easy:s}
	\end{subfigure}
	$\qquad$
    \begin{subfigure}[t]{0.42 \linewidth}
		\centering
		\includegraphics[page=2]{easyInstances}
		\caption{A geometry-free instance for \poleaders: 
		no site lies inside the \poarea of another~site.}
		\label{fig:easy:po}
	\end{subfigure}
  \caption{In a geometry-free instance the leaf order~$\pi$ fully determines
  if any two leaders cross.}
  \label{fig:easy}
\end{figure}

Conversely, consider a site $p_k$ that lies inside the \sarea of $p_i$.
Whether the leaders $s_i$ and~$s_k$ cross depends on the placement 
of the leaves $\ell_i$ and $\ell_k$ in a more complicated way than just their relative order:
$s_i$ might pass left or right of $p_k$
and it is therefore more complicated to determine whether $s_i$ and $s_k$ cross.
In this case, we call the pair $\croc{p_i, p_k}$ \emph{undecided}.
See~\cref{fig:dependencies}, where $p_1$ is undecided with respect to $p_2$.

\begin{figure}[bth]
  \centering
	\includegraphics[page=4]{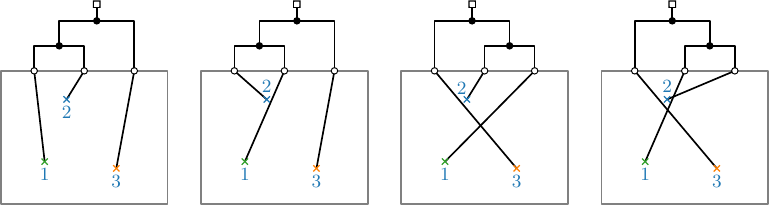}
	\caption{Drawings of the same geophylogeny with four different leaf orders.
	Note that~$s_1$ and~$s_3$ cross 
	if and only if~$\ell_3$ is left of~$\ell_1$.
	On the other hand, whether~$s_1$ and~$s_2$ cross or not depends more specifically on the positions of~$\ell_1$ and~$\ell_2$.}
  \label{fig:dependencies}
\end{figure}

Analogously, for \poleaders, let the \emph{\poarea} of~$p_i$ be the rectangle
that spans horizontally from position~$1$ to position~$n$
and vertically from~$p_i$ to the top of~$R$; see \cref{fig:easy:po}.
A pair $\croc{p_i, p_j}$ of sites is \emph{geometry free}
if $p_i$ does not lie in the \poarea of $p_j$ or vice versa.
A pair $\croc{p_i, p_k}$ of sites is called \emph{undecided},
if $p_k$ lies in the \poarea of $p_i$.
Note that the pairs are ordered and that, for an undecided pair $\croc{p_i, p_k}$, the pair $\croc{p_k, p_i}$ is not undecided for \sleaders
but can be undecided for \poleaders if $\y(p_i) = \y(p_k)$.

We call a geophylogeny \emph{geometry free} (for \sorpoleaders) if all pairs of sites are geometry free.
While it seems unlikely that a geophylogeny is geometry free for \poleaders in practice,
it is not entirely implausible for \sleaders:
for example, researchers may take their samples along a coastline, a river, or a valley,
in which case the sites may lie relatively close to a line.
Orienting the map such that this line is horizontal might then result in a geometry-free instance.
For example, Lazzari, Becerro, Sanabria-Fernandez, and Martín-López~\cite{exampleCoast} considered coastal regions of Andalusia
and their geophylogenies (with dendrograms) could be made geometry-free by rotating the map slightly.
Furthermore, unless two sites share an x-coordinate,
increasing the vertical distance between the map and the tree
eventually results in a geometry-free drawing for \sleaders;
however, the required distance might be impractically large.

Next we show that the number of leader crossings in a geometry-free drawing
can be minimized efficiently using Fernau et~al.'s~\cite{FKP10} algorithm for the \textsc{One-Sided Tanglegram}~problem.
\pagebreak[4]

\begin{proposition} \label[proposition]{clm:easy}
Given a geometry-free geophylogeny $G$ on $n$ taxa, 
a drawing $\Gamma$ with the minimum number of leader crossings 
can be found in~$\Oh(n \log n)$~time, for both \sandpoleaders.
\end{proposition}
\begin{proof} 
To use Fernau et~al.'s~\cite{FKP10} algorithm, 
we transform $G$ into a \emph{one-sided tanglegram}~$(T_{\text{fix}}, T_{\text{vari}})$
that is equivalent in terms of crossing to~$\Gamma$; see~\cref{fig:easy:solving}.
We take the sites~$P$ as the leaves of the tree~$T_{\text{fix}}$ with fixed embedding 
and embed it such that the points are ordered from left to right;
the topology of~$T_{\text{fix}}$ is arbitrary.
As the tree~$T_{\text{vari}}$ with variable embedding, we take the phylogenetic tree~$T$.

If~$\Gamma$ uses \sleaders, then we assume that the sites of~$G$ are indexed from left to right.
If~$\Gamma$ uses \poleaders, we define an (index) order on~$P$ as follows.
Let~$p_i$ be a site and~$p_j$ a site to the right of it; consider the leader that connects $p_i$ to leaf position~$1$ and the leader that connects $p_j$ to leaf position $n$.
If these leaders cross we require that $i$ is after $j$, otherwise it must be before $j$. 
Note that this implies that $p_i$ and $p_j$ are either both left of position 1 or both right of position $n$. 
(It is easily shown that this defines an order; see also \cref{fig:easy:po}.)

Let~$\pi'$ be a leaf order of~$T_{\text{vari}}$.
Further let~$s'_i$ denote the connection of the leaf corresponding to~$p_i$ 
in~$T_{\text{fix}}$ and the leaf~$\ell_i$ in~$T_{\text{vari}}$.
Note that two connections~$s'_i$ and~$s'_j$ with $i<j$
cross in the tanglegram if and only if~$\pi'(\ell_i) > \pi'(\ell_j)$.

\begin{figure}[tbh]
  \centering
	\includegraphics[page=3]{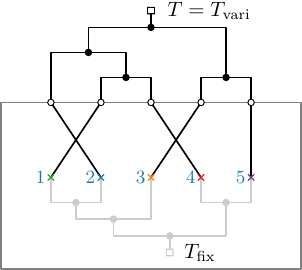}
	\caption{A geometry-free geophylogeny and a one-sided tanglegram~$(T_{\text{fix}}, T_{\text{vari}})$ 
	that have the same combinatorics (in terms of leader crossings) 
	as the two geometry-free instances in \cref{fig:easy}.}
  \label{fig:easy:solving}
\end{figure}

Since~$G$ is geometry free, the crossings in the tanglegram correspond one-to-one 
with those in the geophylogeny drawing with leaf order~$\pi'$;
see again \cref{fig:easy,fig:easy:solving}. 
Hence, the number of crossings of~$(T_{\text{fix}}, T_{\text{vari}})$ 
can be minimized in~$\Oh(n \log n)$~time using an algorithm of Fernau et~al.\ \cite{FKP10}.
The resulting leaf order for~$T_{\text{vari}}$ then also minimizes the number of leader crossings in~$\Gamma$.
\end{proof}

\subsection{FPT Algorithm} % - - - - - - - - - - - -
\label{sec:fpt}
In practice, most geophylogenies are not geometry free,
yet some drawings with \sleaders might have only few sites inside the \sarea of other sites.
Capturing this with a parameter $k$, we can develop an FPT algorithm, that is,
an algorithm that runs in $f(k)p(n)$ time where $f$ is a computable function that only depends on $k$ and $p$ is a polynomial function (see also Niedermeier~\cite{Nie06}).
The idea is as follows. Suppose we use \sleaders and there is exactly one undecided pair $\croc{p_i, p_j}$,
i.e.\ $p_j$ lies inside the \sarea of~$p_i$; see \cref{fig:fpt:idea:geophylo}.
For a particular leaf order, we say the leader $s_i$ \emph{lies left (right)} of $p_j$
if a horizontal ray that starts at $p_j$ and goes to the left (right) intersects $s_i$;
conversely, we say that $p_j$ \emph{lies right (left)} of $s_i$. 

Suppose now that we restrict~$s_i$ to lie left of~$p_j$ 
(as $s_2$ lies left of $p_3$  in \cref{fig:fpt:idea:right}).
This restricts the possible positions for $\ell_i$ and
effectively yields a \emph{restricted} geometry-free geophylogeny.
The idea for our FPT algorithm is thus to use the algorithm from \cref{clm:easy} on restricted geometry-free instances
obtained by assuming that~$s_i$ lies to the left or to the right of~$p_j$; see again \cref{fig:fpt:idea}.
In particular, we extend Fernau et~al.'s dynamic programming algorithm~\cite{FKP10} to handle
\emph{restricted} one-sided tanglegrams at a cost in runtime.

\begin{figure}[tbh]
  \centering
  	\begin{subfigure}[t]{0.3 \linewidth}
		\centering
		\includegraphics[page=1]{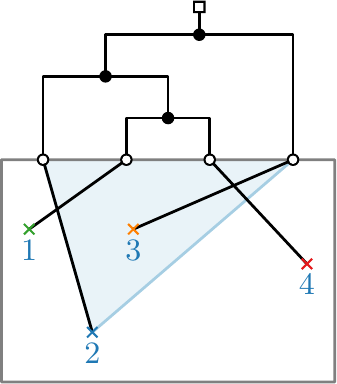}
		\caption{A non-geometry-free geophylogeny since~$p_3$
		lies in the \sarea of~$p_2$.}
		\label{fig:fpt:idea:geophylo}
	\end{subfigure}
	\hfill
    \begin{subfigure}[t]{0.3 \linewidth}
		\centering
		\includegraphics[page=2]{fpt}
		\caption{Restricted geometry-free geophylogeny where we require that~$p_3$ lies right of~$s_2$. 
		Thus~$\ell_2$ cannot be at positions 3 or~4.}
		\label{fig:fpt:idea:right}
	\end{subfigure}
	\hfill
	\begin{subfigure}[t]{0.3 \linewidth}
		\centering
		\includegraphics[page=3]{fpt}
		\caption{Restricted geometry-free geophylogeny where we require that~$p_3$ lies left of~$s_2$. 
		Thus~$\ell_2$ cannot be at positions 1 or~2.}
		\label{fig:fpt:idea:left}
	\end{subfigure}
  \caption{We transform the non-geometry-free geophylogeny~$G$
  into a restricted geometry-free geophylogeny by deciding whether $p_3$ lies left or right of $s_2$.}
  \label{fig:fpt:idea}
\end{figure}

\begin{lemma} \label{clm:tricky}
The number of connection crossings in a restricted one-sided tanglegram~$\cT$ on $n$ leaves
can be minimized in $\Oh(n^3)$ time. 
\end{lemma}
\begin{proof}
Let $\cT = (T_{\text{fix}}, T_{\text{vari}})$; we write $T$ for $T_{\text{vari}}$.
Let $x$ and $y$ be the children of a vertex $v$ of~$T$.
Fernau et~al.'s algorithm would compute the number of crossings $\cro{x,y}$ and $\cro{y,x}$
between the connections of $T(x)$ and the connections of $T(y)$
for when $x$ is the left or right child of $v$, respectively, in $\Oh(n(x) + n(y))$ time.
For an unrestricted one-sided tanglegram, 
this can be done independent of the positions of $T(x)$ and $T(y)$.
For $\cT$ however this would not take into account the forbidden positions of leaves.
Hence, as in our algorithm from \cref{clm:internal},
we add the position of the leftmost leaf of $T(v)$ as additional parameter in the recursion.
This adds a factor of $n$ to the running time
and thus, forgoing Fernau et~al.'s data structures, results in a total running time in~$\Oh(n^3)$. 
\end{proof}

Before describing an FPT algorithm based on restricted geometry-free geophylogenies,
let us consider the example from \cref{fig:dependencies} again.
There the drawing $\Gamma$ has three sites~$p_1$, $p_2$, $p_3$ 
where~$p_2$ lies in the \sarea of both~$p_1$ and~$p_3$.
We can get four restricted geometry-free geophylogenies by
requesting that~$p_2$ lies to the left or to the right of~$s_1$ and of~$s_3$.
Here one of the instances, $G'$, stands out,
namely where~$p_2$ lies to the left of~$s_1$ and to the right of~$s_3$; see \cref{fig:fpt:conflict:bad}.
In the restricted one-sided tanglegram $\cT'$ corresponding to $G'$, 
we would want~$p_2$ left of~$p_1$ and right of~$p_3$.
This stands in conflict with~$p_1$ being left of~$p_3$ based on their indices.
We thus say $p_1$, $p_2$, and $p_3$ 
form a \emph{conflicting triple} $\croc{p_1, p_2, p_3}$, which we resolve as follows.
Note that~$s_1$ and~$s_3$ cross for any valid leaf order for $G'$.
We thus use the order~$p_3$,~$p_2$,~$p_1$ for $\cT'$ (see \cref{fig:fpt:conflict:tangle})
and, since $\cT'$ does not contain the crossing of~$s'_1$ and~$s'_3$,
we add one extra crossing to the computed solution.
A conflicting triple for drawings with \poleaders is defined analogously.

\begin{figure}[tbh]
  \centering
    \begin{subfigure}[t]{0.41 \linewidth}
		\centering
		\includegraphics[page=4]{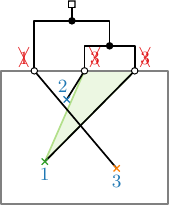}
		\caption{Restricted geometry-free geophylogeny with conflicting triple $\croc{p_1, p_2, p_3}$.
		Note that $s_1$ and $s_3$ cross for any valid leaf order.}
		\label{fig:fpt:conflict:bad}
	\end{subfigure}
	\quad
    \begin{subfigure}[t]{0.41 \linewidth}
		\centering
		\includegraphics[page=5]{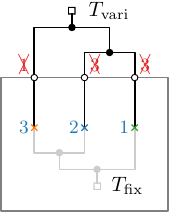}
		\caption{Restricted one-sided tanglegram $\cT'$ with order $p_3$, $p_2$, $p_1$.
		Note that $s'_1$ and $s'_3$ do not cross for any valid leaf order.}
		\label{fig:fpt:conflict:tangle}
	\end{subfigure}
  \caption{A conflicting triple in a restricted geometry-free instance requires a different order
  in the corresponding restricted one-sided tanglegram and storing the ``lost'' crossing.}
  \label{fig:fpt:conflict}
\end{figure}

\begin{theorem} \label{clm:fpt}
Given a geophylogeny $G$ on $n$ taxa and with $k$ undecided pairs of sites,
a drawing~$\Gamma$ of~$G$ with minimum number of crossings can be computed in~$\Oh(2^k \cdot (k^2 + n^3))$ time, 
for both \sandpoleaders.
\end{theorem}
\begin{proof}
Our FPT algorithm converts~$G$ into up to~$2^k$ restricted geometry-free instances,
solves the corresponding restricted one-sided tanglegrams with \cref{clm:tricky}, 
and then picks the leaf order that results in the least leader crossings for~$\Gamma$.
Therefore, for each undecided pair $\croc{p_i, p_j}$, 
the algorithm tries routing~$s_i$ either to the left of or to the right of~$p_j$.
Since there are~$k$ such pairs, there are~$2^k$ different combinations.
However, for some combinations a drawing might be over restricted and no solution exists.

To go through all possible combinations, we branch for each of the~$k$ undecided pairs $\croc{p_i, p_j}$,
whether~$s_i$ is to the left or to the right of~$p_j$.
Let $\omega$ encode one sequence of $k$ such decisions.
Below we show how to construct the restricted geometry-free instances~$G_\omega$ 
and the corresponding restricted one-sided tanglegram~$\cT_\omega$ in~$\Oh(k^2 + n^2)$ time.
Since the number of crossings in the restricted geometry-free drawing can 
then be minimized in~$\Oh(n^3)$ time with \cref{clm:tricky},
the claim on the running time follows.  

In order to construct~$\Gamma_\omega$ efficiently, 
we keep track of the positions where a leaf~$\ell_i$, for~$i \in  \set{1, \ldots, n}$,  
can be placed with an interval~$[a_i, b_i]$; 
at the start we have~$a_i = 0$ and~$b_i = n$.
Suppose that when going through the~$k$ pairs and for the current~$\omega$,
we get that~$s_i$ becomes restricted by, say, having to be left of a site~$p_j$. 
Then we compute the rightmost position where~$\ell_i$ could be placed
and update~$b_i$ accordingly. This can both be done in constant time.
If at any moment~$a_i > b_i$, then the drawing for~$\omega$ 
is over restricted and there is no viable leaf order. 
We then continue with the next combination.
Otherwise, after all~$k$ pairs, we have restricted~$G$ to~$G_\omega$ in~$\Oh(k)$ time.

Next, we explain how to find an order of~$P$ for~$\cT_\omega$
that corresponds to~$G_\omega$.
In particular, we have to show that resolving all conflicting triples as described above
in fact yields an order of~$P$.
To this end, let~$K$ be the complete graph with vertex set~$P$.
(We assume again the same order on the sites as in \cref{clm:easy}.)
For any two sites~$p_i, p_j \in P$ with~$i < j$,
we orient~$\set{p_i, p_j}$ as~$(p_j, p_i)$
if~$\croc{p_i, p_j}$ is an undecided pair and
$s_i$ is right of~$p_j$ in~$G_\omega$; otherwise we orient it as~$(p_i, p_j)$.
We then check whether any pair of undecided pairs forms a conflicting triple.
For any conflicting triple~$\croc{p_i, p_j, p_l}$ that we find, 
we \emph{reorient} the edge between~$p_i$ and~$p_l$ to~$(p_l, p_i)$.
We claim that~$K$ is acyclic (and prove it below).
Therefore we can use a topological order of~$K$ as order for~$\cT_\omega$.
For~$i \in \set{1, \ldots, n}$,
we set the dynamic programming values for leaf~$\ell_i$ 
at all positions in~$\cT_\omega$ outside of~$[a_i, b_i]$ to infinity.
We can find all conflicting triples in~$\Oh(k^2)$ time,
construct and orient~$K$ in~$\Oh(n^2)$ time, 
and initialize~$\cT$ in~$\Oh(n^2)$ time.

Lastly, we show that~$K$ is indeed acyclic after resolving all conflicting triples.
Suppose to the contrary that there is a directed cycle~$C'$ in~$K$.
Since the underlying graph of~$K$ is the complete graph, 
there is then also a directed triangle~$C$ in~$K$.
To arrive at a contradiction, we show that~$C$ cannot have 0, 1, 2, or 3 reoriented edges.
Let~$C$ be on~$p_i$,~$p_j$,~$p_l$ with edges~$(p_i, p_j)$,~$(p_j, p_l)$, and~$(p_l, p_i)$.

\textbf{$\bm{ C }$ contains 0 reoriented edges:}
Since~$p_i$,~$p_j$, and~$p_l$ do not form a conflicting triple,
an easy geometric case distinction shows that this is not geometrically realizable.
For example, if, say,~$p_i$ is lower than~$p_j$ and~$p_j$ is lower than~$p_l$, 
then~$p_j$ is right of~$s_i$ and~$p_l$ is right of~$s_j$.
However, then~$p_l$ cannot be left of~$s_i$ and~$p_i$ cannot be right of~$s_l$; see \cref{fig:fpt:triangle:zero}.

\begin{figure}[tbh]
  \centering
    \begin{subfigure}[t]{0.31 \linewidth}
		\centering
 		\includegraphics[page=6]{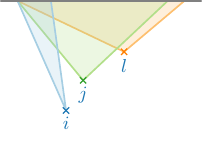}
		\caption{If~$C$ contains no reoriented edge
		and $p_i$, $p_j$, $p_l$ are not a conflicting triple,
		then there is no geometric realization respecting $C$;
		here $(p_l, p_i)$ is not realized.}
		\label{fig:fpt:triangle:zero}
	\end{subfigure}
	\hfill
    \begin{subfigure}[t]{0.31 \linewidth}
		\centering
 		\includegraphics[page=7]{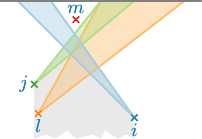}
		\caption{If~$C$ contains one reoriented edge~$(p_i, p_j)$,
		there is no placement of $p_l$ respecting $C$ 
		where $p_l$ is not part of a conflicting triple with $p_m$ and $p_i$ (as here) or $p_j$.}
		\label{fig:fpt:triangle:one}
	\end{subfigure}
	\hfill
    \begin{subfigure}[t]{0.31 \linewidth}
		\centering
 		\includegraphics[page=8]{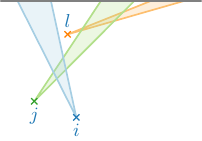}
		\caption{If~$C$ contains two reoriented edges~$(p_i, p_j)$ and~$(p_j, p_l)$,
		we get that~$p_l$ must lie right of~$s_i$.}
		\label{fig:fpt:triangle:two}
	\end{subfigure}
  \caption{Cases for the proof of \cref{clm:fpt} where the directed triangle~$C$ 
  with edges~$(p_i, p_j)$,~$(p_j, p_l)$, and~$(p_l, p_i)$
  is supposed to contains no, one, or two reoriented edges. 
  Colored cones represent a hypothetical possible range for the leader of the corresponding site.}
  \label{fig:fpt:triangle}
\end{figure}

\textbf{$\bm{ C }$ contains 1 reoriented edge:}
Suppose that~$(p_i, p_j)$ has been reoriented as part of a conflicting triple with~$p_m$.
We show that $p_m$ also lies in the \sarea (\poarea) of $p_l$:
If $p_i$ or $p_j$ lies in the \sarea (\poarea) of $p_l$, then so does $p_m$ by transitivity.
Assuming otherwise, $\croc{p_l, p_j}$ and $\croc{p_l, p_i}$ cannot be undecided pairs.
Therefore $(p_l, p_i)$ implies that either $\croc{p_i, p_l}$ is an undecided pair, meaning $p_l$ is left of $s_i$, or that $l < i$, meaning $p_l$ lies left of $p_i$.
Analogously, $(p_j, p_l)$ implies that $p_l$ lies right of $s_j$ or right of $p_j$ (or both).
Taken together, $p_l$ lies ``between'' $s_j$ and $s_i$ as well as below the crossing of $s_i$ and $s_j$; see \cref{fig:fpt:triangle:one}.
On the other hand, $p_m$ lies lies ``between'' $s_i$ and $s_j$ as well as above the crossing.
Therefore, $p_m$ lies in the \sarea (\poarea) of $p_l$.
However, then based on the orientation of $(p_j, p_l)$, and~$(p_l, p_i)$,
we get that $p_l$ must form a conflict triple with $p_m$ and either $p_i$ or $p_j$.
This stands in contradiction to $C$ containing only one reoriented edge.

\textbf{$\bm{ C }$ contains 2 reoriented edges:}
Suppose that~$(p_i, p_j)$ and~$(p_j, p_l)$ have been reoriented.
We then know that~$s_i$ and~$s_j$ as well as~$s_j$ and~$s_l$ definitely cross.
Therefore,~$p_i$ lies right of~$s_j$ (or the line through~$s_j$) 
and that~$p_j$ lies left of~$s_i$ (or the line through~$s_i$).
Analogously,~$p_j$ lies right of~$s_j$ (or the line through~$s_l$), 
and that~$p_l$ lies left of~$s_j$ (or the line through~$s_j$).
Since~$(p_l, p_i)$ has not been reoriented,
we know that~$s_i$ and~$s_l$ do not necessarily need to cross.
For~$s_l$ to cross~$s_j$ but not~$s_i$, we get that~$p_l$ can only lie right of~$s_i$; see \cref{fig:fpt:triangle:two}. 
This stands in contradiction to the orientation of~$(p_l, p_i)$.

\textbf{$\bm{ C }$ contains 3 reoriented edges:}
Since all three edges of~$C$ have been reoriented, 
this is geometrically equivalent to the first case and thus not realizable.

This concludes the proof that there is no directed triangle in $K$ and hence $K$ is acyclic.
Our FPT algorithm can thus process each of the $2^k$ words in $\Oh(k^2 + n^2)$ time.
\end{proof}

Note that a single site can lie in the \sarea of every other site,
for example, this is likely for a site that lies very close to the top of the map. 
Furthermore, there can be $\Oh(n^2)$ undecided pairs.
In these cases, the running time of the FPT algorithm becomes~$\Oh(2^n n^2)$ or even~$\Oh(2^{(n^2)} n^4)$.
However, a brute-force algorithm that tries all $2^{n-1}$ embeddings of $T$ 
and computes for each the number of leader crossings in $\Oh(n^2)$ time,
only has a running time in~$\Oh(2^n n^2)$.

\subsection{Integer Linear Programming} % - - - - - - - - - - - -
\label{sec:ilp} 
As we have seen above, the problem of minimizing the number of leader crossings
in drawings of geophylogenies is NP-hard and 
the preceding algorithms can be expected to be impractical on realistic instances. 
We now provide a practical method to exactly solve instances of moderate size using integer linear programming (ILP).

For the following ILP, we consider an arbitrary embedding of the tree as \emph{neutral} and describe all embeddings 
in terms of which internal vertices of $T$ are rotated with respect to this neutral embedding,
i.e.\ for which internal vertices to swap the left-to-right order of their two children. 
For two sites $p_i$ and $p_j$, we use $p_i \prec p_j$ to denote that $\ell_i$ is left of $\ell_j$ in the neutral embedding.
Let $U$ be the set of undecided pairs, that is, all ordered pairs $(p,q)$ where $q$ lies inside the \sarea of $p$;
note that these are ordered pairs.
We further assume that position $i$ corresponds to the x-coordinates~$i$.

\paragraph*{Variables and Objective Function.}
The program has three groups of binary variables that describe the embedding and crossings.
\begin{description}
    \item[$\rho_u \in \set{0, 1} \ \forall u \in I(T)$.]
	Rotate internal vertex $u$ if $\rho_u = 1$ and keep its neutral embedding if $\rho_u = 0$.
	Note that rotating the lowest common ancestor of leaves $\ell_i$ and $\ell_j$ is the only way to flip their order, so
	for convenience we write $\rho_{ij}$ to mean $\rho_{\lca{i,j}}$.
	Note, however, that an internal vertex can be the lowest common ancestor of multiple pairs of leaves.

	\item[$d_{pq} \in \set{0, 1} \ \forall (p,q) \in U$.]
	For each undecided pair $(p,q)$: the leader for $p$ should pass to the left of the site $q$ if $d_{pq}=0$ and to the right if $d_{pq}=1$.
	This is well-defined since the pair is undecided.
	
	\item[$\chi_{pq} \in \set{0, 1} \ \forall p,q \in P,\, p < q$.] 
	For each pair of distinct sites: the leaders of $p$ and $q$ are allowed to cross if $\chi_{pq}=1$ and are not allowed to cross if $\chi_{pq}=0$.
\end{description}
There is no requirement that non-crossing pairs have $\chi_{pq} = 0$, but that will be the case in an optimal solution:
to minimize the number of crossings, we minimize the sum over all $\chi_{pq}$.

\paragraph*{Constraints.}
We handle geometry-free pairs and undecided pairs separately.

Consider a geometry-free pair of sites: if the leaders cross in the neutral embedding, we must either allow this, or rotate the lowest common ancestor.
Conversely, if they do not cross neutrally, yet we rotate the lowest common ancestor, then we must allow their leaders to cross.
Call these sets of pairs $F_\text{rotate}$ and $F_\text{keep}$ respectively, for how to prevent the crossing.
\begin{equation}
	\chi_{ij} + \rho_{ij} \geq 1 \quad\quad \forall (i,j) \in F_\text{rotate}
\end{equation}
\begin{equation}
	\chi_{ij} - \rho_{ij} \geq 0 \quad\quad \forall (i,j) \in F_\text{keep}
\end{equation}

\pagebreak[4]
For undecided pairs $(p,q)$, a three-way case distinction on $[p \prec q]$, $\rho_{pq}$, and $d_{pq}$ reveals the following geometry:
\begin{itemize}
	\item pairs with $p \prec q$ have crossing leaders if and only if $\rho_{pq} + d_{pq} = 1$; 
	\item pairs with $p \succ q$ have crossing leaders if and only if $\rho_{pq} + d_{pq} \neq 1$.
\end{itemize}
Recall that we do not force $\chi$ to be zero if there is no intersection, 
only that it is $1$ if there {\em is} an intersection.
We implement these conditions in the ILP as follows.
Let $U_\text{left} \subseteq U$ be the undecided pairs with $p \prec q$.
\begin{equation}
    \rho_{pq} - d_{pq} \leq \chi_{pq} \quad\quad \forall (p,q) \in U_\text{left}
\end{equation}
\begin{equation}
    d_{pq} - \rho_{pq} \leq \chi_{pq} \quad\quad \forall (p,q) \in U_\text{left}
\end{equation}
Conversely, let $U_\text{right} \subseteq U$ be the undecided pairs with $p \succ q$.
\begin{equation}
    \rho_{pq} + d_{pq} - 1 \leq \chi_{pq} \quad\quad \forall (p,q) \in U_\text{right}
\end{equation}
\begin{equation}
    1 - \rho_{pq} - d_{pq} \leq \chi_{pq} \quad\quad \forall (p,q) \in U_\text{right}
\end{equation}
Finally, we must ensure that each leader $s_i$ respects the $d$ variables:
the line segment from $p_i$ to~$\ell_i$ must pass by each other site in the \sarea on the correct side.
By their definition, this does not affect geometry-free pairs, but it remains to constrain the leaf placement for undecided pairs.

Observe that the $\rho$ variables together fix the leaf order, since they fix the embedding of~$T$.
Let $L_i(\rho)$ be the function that gives the x-coordinate of $\ell_i$ given the $\rho$ variables.
Note that $L_i$ is linear in each of the $\rho$ variables: 
rotating an ancestor of $\ell_i$ shifts its location leaf by a particular constant, and rotating a non-ancestor does not affect it.

\begin{figure}[tbh]
	\centering
	  \begin{subfigure}[t]{0.29 \linewidth}
		\centering
		\includegraphics[page=1]{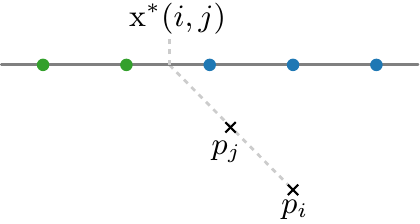}
		\caption{\sleaders{}}
		\label{fig:xstar:sub}
	  \end{subfigure}
	  \hspace{1cm}
	  \begin{subfigure}[t]{0.29 \linewidth}
		\centering
		\includegraphics[page=2]{ilpVarD}
		\caption{\poleaders{}}
		\label{fig:xstar:sub}
	  \end{subfigure}
	\caption{Defining $\x^*(i,j)$ by extending a leader from $s_i$ through $s_j$
	This partitions the leaf positions into those where the leader goes left of $s_j$ and those where it goes right.}
	\label{fig:xstar:main}
  \end{figure}

For an undecided pair $(p_i,p_j)$, consider a leader starting at $p_i$ and extending up through $p_j$:
for \sleaders{} this is the ray from $p_i$ through $p_j$, for \poleaders{} this is the vertical line through $p_j$.
Let $\x^*(i,j)$ be the x-coordinate of where this extended leader intersects the top of the map and note that this is a constant; see~\cref{fig:xstar:sub}.
If $d_{ij} = 0$, then $\ell_i$ must be to the left of this intersection; if $d_{ij} = 1$, it must be to the right.
We model this in the ILP with two constraints and the \emph{big-M method}, where it suffices to set $M = n$.
\begin{equation}
	L_i(\rho) - d_{ij}M       \leq \x^*(i,j) \quad\quad \forall (p_i,p_j) \in U
\end{equation}
\begin{equation}
	L_i(\rho) + (1 - d_{ij})M \geq \x^*(i,j) \quad\quad \forall (p_i,p_j) \in U
\end{equation}
This completes the ILP.

The number of variables and constraints are both quadratic in $n$.
Just counting the $\chi$ variables already gives this number, but we note that
in particular the number of undecided pairs leads to additional variables (and seemingly more complicated constraints).

\subsection{Heuristics} % - - - - - - - - - - - -
\label{sec:heuristics} 
Since the ILP from the previous section can be slow in the worst case and requires advanced solver software,
we now suggest a number of heuristics.

\paragraph{Bottom-Up.}
First, we use a dynamic program similar to the one in \cref{sec:internal} 
and commit to an embedding for each subtree while going up the tree.
At this point we note that counting the number of crossings is not 
a leaf additive quality measure in the sense of \cref{sec:internal}.
However, \cref{eq:dp} does enable us to introduce an additional cost based on 
where an entire subtree is placed and where its sibling subtree is placed 
--~just not minimized over the embedding of these subtrees.
More precisely, for an inner vertex~$v$ of~$T$ with children~$x$ and~$y$,
let~$C(x,y,i)$ be the number of crossings between~$T(x)$ and~$T(y)$ 
when placed starting at position~$i$ and~$i+n(x)$ respectively; this can be computed in~$\Oh(n(v)^2)$ time.
Note that this ignores any crossings with leaders from other parts of the tree.
With base case~$H(\ell, i) = 0$ for every leaf~$\ell$,
we use
\begin{equation*}
	H(v, i) = \min\set{\quad H(x, i) + H(y, i + n(x)) + C(x,y,i),\quad H(y, i) + H(x, i + n(y)) + C(y,x,i)\quad }
\end{equation*}
\noindent
to pick a rotation of~$T(v)$.
Since $H$ can be evaluated in~$\Oh(n^2)$ time, the heuristic runs in~$\Oh(n^4)$ time total.
The example in \cref{fig:heuristic:notOpt} demonstrates that this does not minimize the total number of~crossings.

\begin{figure}[tbh]
  \centering
  	\begin{subfigure}[t]{0.29 \linewidth}
		\centering
		\includegraphics[page=1]{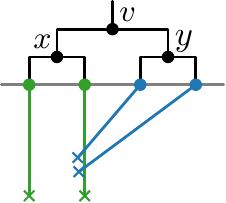}
		\caption{When~$T(x)$ and~$T(y)$ each have zero crossings, \ldots}
		\label{fig:heuristic:notOpt:left}
	\end{subfigure}
	\quad
    \begin{subfigure}[t]{0.29 \linewidth}
		\centering
		\includegraphics[page=2]{heuristicNotOptimal}
		\caption{\ldots then~$T(v)$ has two crossings.}
		\label{fig:heuristic:notOpt:right}
	\end{subfigure}
	\quad
	\begin{subfigure}[t]{0.29 \linewidth}
		\centering
		\includegraphics[page=3]{heuristicNotOptimal}
		\caption{The optimal leaf order for~$T(v)$ has one crossing in~$T(x)$.}
		\label{fig:heuristic:notOpt:opt}
	\end{subfigure}
  \caption{The bottom-up heuristic is not always optimal:
  combining the locally best leaf orders for~$T(x)$ and~$T(y)$ might 
  not result in the minimum number of leader crossings for~$T(v)$.}
  \label{fig:heuristic:notOpt}
\end{figure}

\paragraph{Top-Down.}
The second heuristic traverses $T$ from top to bottom (i.e.\ in pre-order)
and chooses a rotation for each inner vertex $v$ based 
on how many leaders would cross the vertical line between the two subtrees of $v$; see \cref{fig:heuristic:topDown}.
More precisely, suppose that $T(v)$ has its leftmost leaf at position~$i$ 
based on the rotations of the vertices above $v$.
For $x$ and $y$ the children of $v$,
consider the rotation of $v$ where $T(x)$ is placed starting at position~$i$ and 
$T(y)$ is placed starting at position $i + n(x)$.
Let $s$ be the x-coordinate in the middle between the last leaf of $T(x)$ and 
the first leaf of $T(y)$.
We compute the number of leaders of $T(v)$ that cross the vertical line at $s$
and for the reserve rotation of $v$;
the smaller result is chosen and the rotation fixed.
This procedure considers each site at most $\Oh(n)$ times and thus runs in $\Oh(n^2)$ time. 

\begin{figure}[tbh]
	\centering
	\begin{subfigure}[t]{.44 \linewidth}
		\centering
		\includegraphics[page=1]{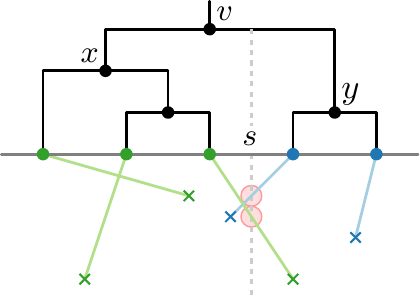}
		\caption{$T(x)$ and $T(y)$ have one site each on the other side
		of the vertical line through $s$.}
	\end{subfigure}
	\hfill
	\begin{subfigure}[t]{.44 \linewidth}
		\centering
		\includegraphics[page=2]{heuristicTopDown}
		\caption{$T(x)$ has one and  $T(y)$ has two sites on the  other side
		of the vertical line through $s$.}
	\end{subfigure}
  \caption{The top-down heuristic tries both rotations of $v$ and 
  here would pick (a).}
  \label{fig:heuristic:topDown}
\end{figure}

\paragraph{Leaf-Additive Dynamic Programming.}
Thirdly, we could optimize any of the quality measures for interior labeling (\cref{sec:internal}).
These measures produce generally sensible leaf orders in quadratic time 
and we may expect the number of leader crossings to be low.

\paragraph{Greedy (Hill Climbing).}
Finally, we consider a hill climbing algorithm that, starting from some leaf order, 
greedily performs rotations that improve the number of crossings.
This could start from a random leaf order, a hand-made one, or from any of the other heuristics.
Evaluating a rotation can be done in $\Oh(n^2)$ time and thus one round through all vertices runs in $\Oh(n^3)$ time.

\section{Experimental Evaluation} % -----------------------
\label{sec:experiments}
In this section, we evaluate the practical performance of our proposed ILP formulation and heuristics on both synthetic and real-world instances.
The experiments aim to assess the solution quality, relative performance, and computational efficiency of the methods.  
Both the code and test data are available online\footnote{GitHub geophylo repository \href{https://www.github.com/joklawitter/geophylo}{\texttt{github.com/joklawitter/geophylo}}}.

\subsection{Test Data}
We use three procedures to generate random instances.
For each type and with 10 to 100 taxa (in increments of 5), we generated 10 instances; we call these the \emph{synthetic instances}.
We stop at $100$ since geophylogeny drawings with more taxa are rarely well readable. Example instances are shown in \cref{fig:data}.
\begin{description}
    \item[Uniform] Place $n$ sites on the map uniformly at random.
	Generate the phylogenetic tree by repeating the following merging procedure.
    Pick an unmerged site or a merged subtree uniformly at random, then pick a second with probability distributed by inverse distance to the first, and merge them;
	as reference point for the sites of a subtree, we take the median coordinate on both~axes.
	\item[Coastline] Initially place all sites equidistantly on a horizontal line, 
	then slightly perturb the x-coordinates. 
	Next, starting at the central site and going outwards in both directions, 
	change the y-coordinate of each site randomly (up to $1.5$ times the horizontal distance) 
	from the y-coordinate of the previous site. Construct the tree as before.
	\item[Clustered] These instances group multiple taxa into clusters.
    First a uniformly random number of sites between three and ten is allocated for a cluster and its center is placed at a uniformly random point on the map.
	Then for each cluster, we place sites randomly in a disk around the center with size proportional to the cluster size.
	Construct $T$ as before, but first for each cluster separately and only then for the whole instance.
\end{description}
\newcommand{\instance}[1]{{\sffamily\bfseries #1}}
In addition, we consider three real world instances derived from published drawings.
\instance{Fish} is a 14-taxon geophylogeny by Williams and Johnson~\cite{realWorldExampleFish}
with 24 undecided pairs (26\% of possible pairs), which could be reduced to 14 by rotating the map.
\instance{Lizards} is a 20-taxon geophylogeny by Jauss et~al.\ \cite{externalExample},
where the sites are mostly horizontally dispersed, resulting in 38 undecided pairs (20\%, see \cref{fig:examples:external}).
\instance{Frogs} is a 64-taxon geophylogeny by Ellepola et~al.\ \cite{realWorldExampleTwo},
where the sites are rather chaotically dispersed on the map;
the published drawing of Frogs uses \sleaders and has over 680 crossings.

\begin{figure}[tbh]
  \centering
	\begin{subfigure}[t]{0.31 \linewidth}
		\centering
 		\includegraphics[width=\linewidth]{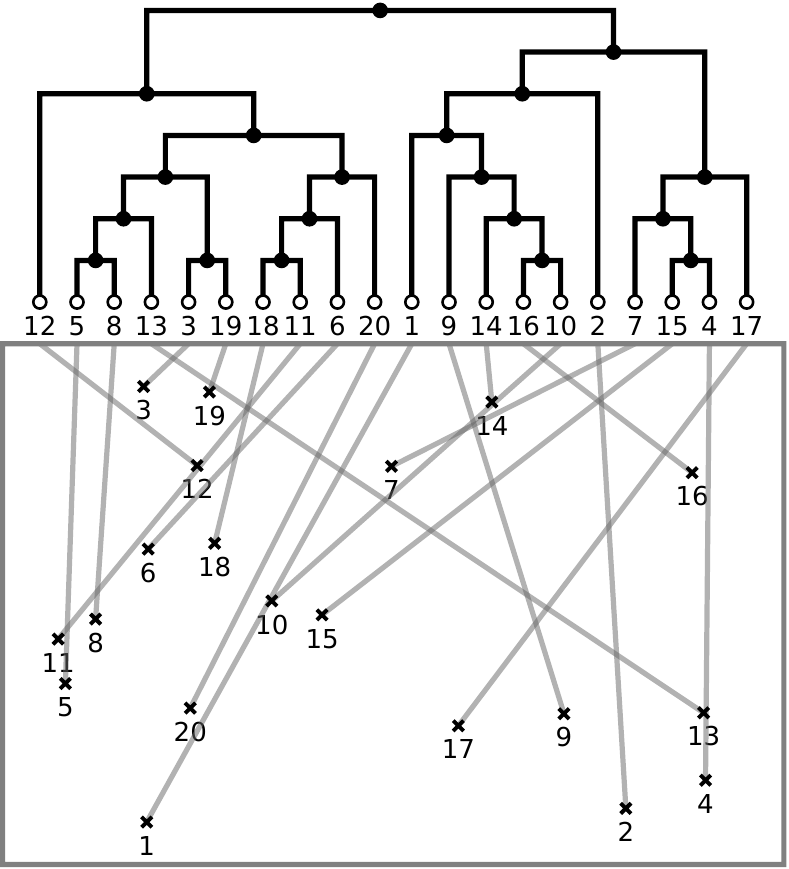}
		\caption{uniform}
		\label{fig:data:uniform}
	\end{subfigure}
	\hfill
    \begin{subfigure}[t]{0.31 \linewidth}
		\centering
 		\includegraphics[width=\linewidth]{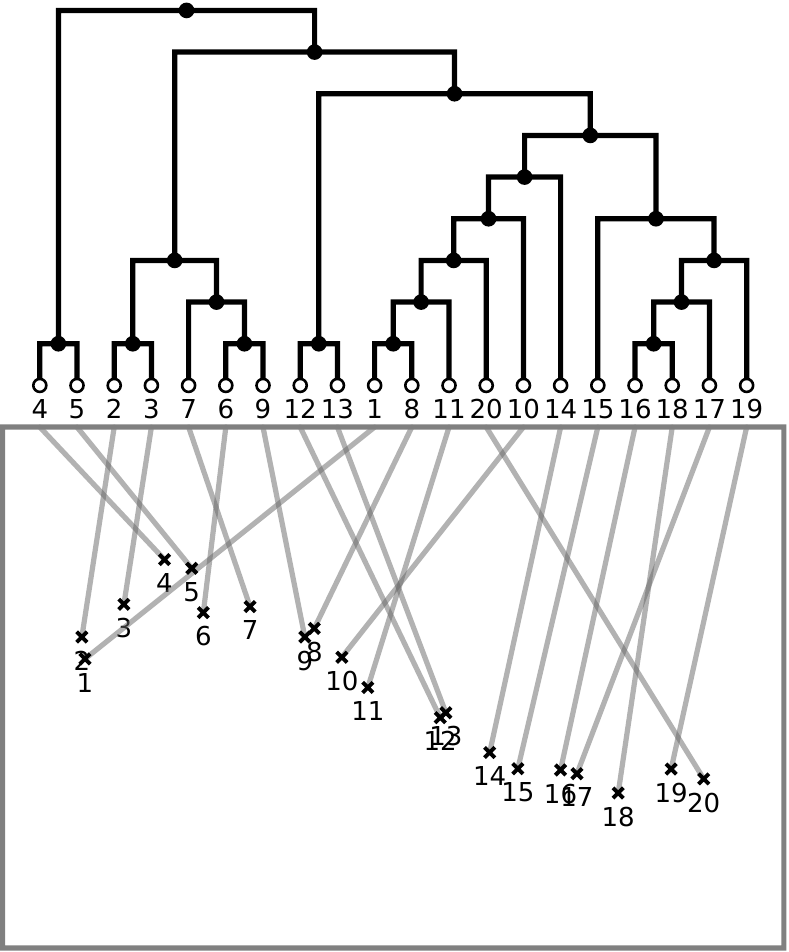}
		\caption{coastline}
		\label{fig:data:coast}
	\end{subfigure}
	\hfill
	\begin{subfigure}[t]{0.31 \linewidth}
		\centering
 		\includegraphics[width=\linewidth]{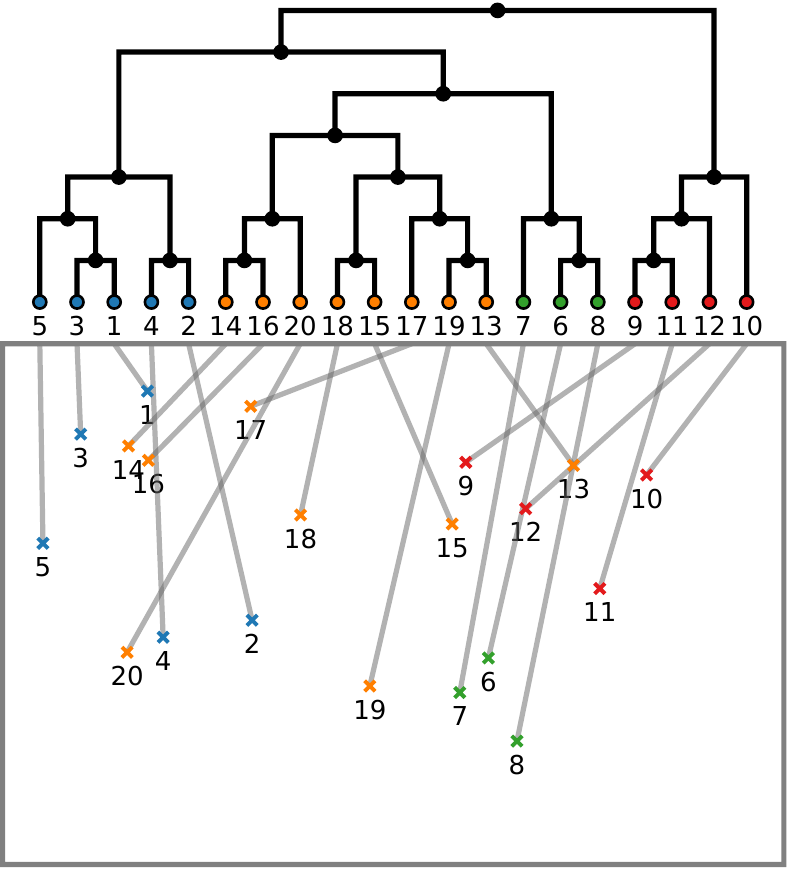}
		\caption{cluster}
		\label{fig:data:cluster}
	\end{subfigure}
  \caption{Examples of generated instances with 20 taxa, here shown with \sleaders. 
  The leaf order was computed with the Greedy hill climber.}
  \label{fig:data}
\end{figure}

\subsection{Experimental Results}
We now describe the main findings from our computational experiments.

\paragraph{The ILP is fairly quick for \sleaders{}.}
Our implementation uses Python to generate the ILP instance and Gurobi~10 to solve it; 
we ran the experiments on a 10-core Apple M1 Max processor.
The Python code takes negligible time; practically all time is spent in the ILP solver. 
As expected, we observe that the running time is exponential in $n$, but only moderately so (\cref{fig:optstats}).
Instances with up to about 50 taxa can usually be solved optimally within a second,
but for Clustered and Uniform instances the ILP starts to get slow at about 100 taxa.
We note that geophylogenies with over $100$ taxa should probably not be drawn with external labeling:
for example, the Frogs instance can be drawn optimally by the ILP in about \unit{0.5}{\second},
but even though this improves the number of crossings from the published 680 to the optimal 609, the drawing is so messy as to be unreadable (\cref{fig:result:external:frog}).
We further observe that Coastline instances are solved trivially fast, 
since with fewer undecided pairs the ILP is smaller and presumably easier to solve.

\paragraph{The ILP is noticeably slower for \poleaders{}.}
Instances with up to $25$ taxa are still drawn comfortably within a second, but at $50$ taxa the typical runtime is over a minute.
We conjecture this is due to the increased number of undecided pairs when working with \poleaders{}.

\begin{figure}
	\centering
    \begin{subfigure}{0.48\linewidth}
	\begin{tikzpicture}
	\begin{axis} [
		table/col sep=semicolon,
        width=\linewidth,
		legend pos=north west,
		ymode=log,
		xlabel={$n$},
		ylabel={Runtime [\second]}
		]
	\addplot [ only marks, mark=+, plot-uni ] table [x=n, y=uni-time] {ilp-stats.csv};
	\addlegendentry{Uniform}
	\addplot [ only marks, mark=x, plot-coast ] table [x=n, y=coast-time] {ilp-stats.csv};
	\addlegendentry{Coastline}
	\addplot [ only marks, mark=o, plot-cluster ] table [x=n, y=cluster-time] {ilp-stats.csv};
	\addlegendentry{Clustered}
	\end{axis}
	\end{tikzpicture}
    \caption{Runtime}
    \end{subfigure}
	\hfill
    \begin{subfigure}{.48\linewidth}
	\begin{tikzpicture}
		\begin{axis} [
			table/col sep=semicolon,
            width=\linewidth,
			legend pos=north west,
			xlabel={$n$},
			ylabel={OPT}
			]
		\addplot [ only marks, mark=+, plot-uni ] table [x=n, y=uni-xing] {ilp-stats.csv};
		\addlegendentry{Uniform}
		\addplot [ only marks, mark=x, plot-coast ] table [x=n, y=coast-xing] {ilp-stats.csv};
		\addlegendentry{Coastline}
		\addplot [ only marks, mark=o, plot-cluster ] table [x=n, y=cluster-xing] {ilp-stats.csv};
		\addlegendentry{Clustered}
		\end{axis}
		\end{tikzpicture}
        \caption{Number of crossings.}
    \end{subfigure}
	\caption{
		Computing optimal \sleader{} drawings using the ILP.
		}
	\label{fig:optstats}
\end{figure}
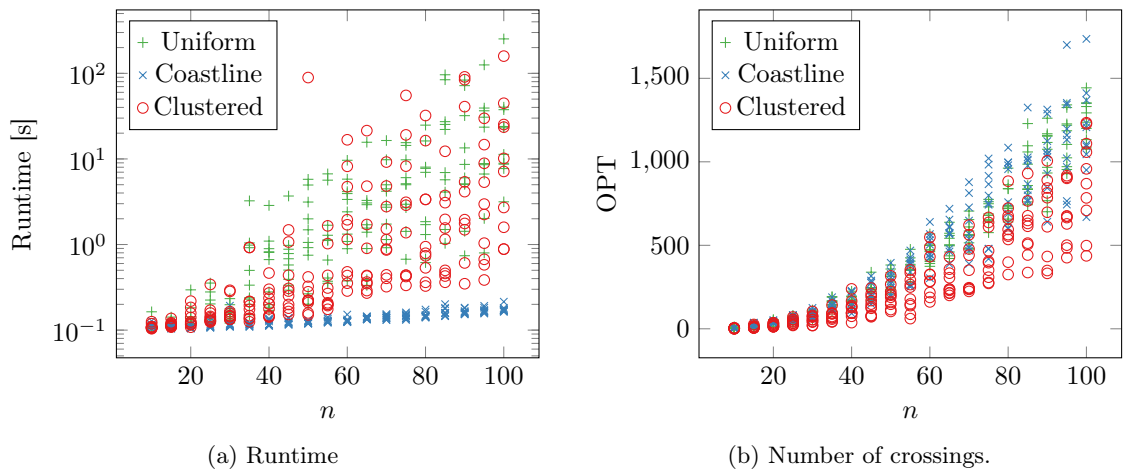

\paragraph{The synthetic instances have a superlinear number of crossings.}
The Clustered instances can be drawn with significantly fewer crossings than Uniform:
this matches our expectation, as by construction there is more correlation between the phylogenetic tree and the geography of the sites.
More surprisingly we find that the Coastline instances require many crossings.
We may have made these geophylogenies too noisy, but this observation does warn of the generally quadratic growth in the number of crossings,
which makes external labeling unsuitable for large geophylogenies unless the geographic correlation is exceptionally good.

\paragraph{The heuristics run instantly and Greedy is often optimal.}
The heuristics are implemented in single-threaded Java code;
we ran the experiments on an average 4-core laptop.
Bottom-Up, Top-Down and Leaf-Additive all run instantly.
Even the Greedy hill climber finds a local optimum in a fraction of a second, both when starting with a random leaf order or from any of the other heuristics.
Of the first three heuristics, Bottom-Up consistently achieves the best results for both \sandpoleaders. 
Comparing the best solution by these heuristics with the optimal drawing (\cref{fig:excessxing}), 
we observe that the number crossings in excess of the optimum increases with the number of taxa, in particular for Uniform and Clustered instances;
Coastline instances are always drawn close to optimally by at least one heuristic.
The Greedy hill climber almost always improves this to an optimal~solution.

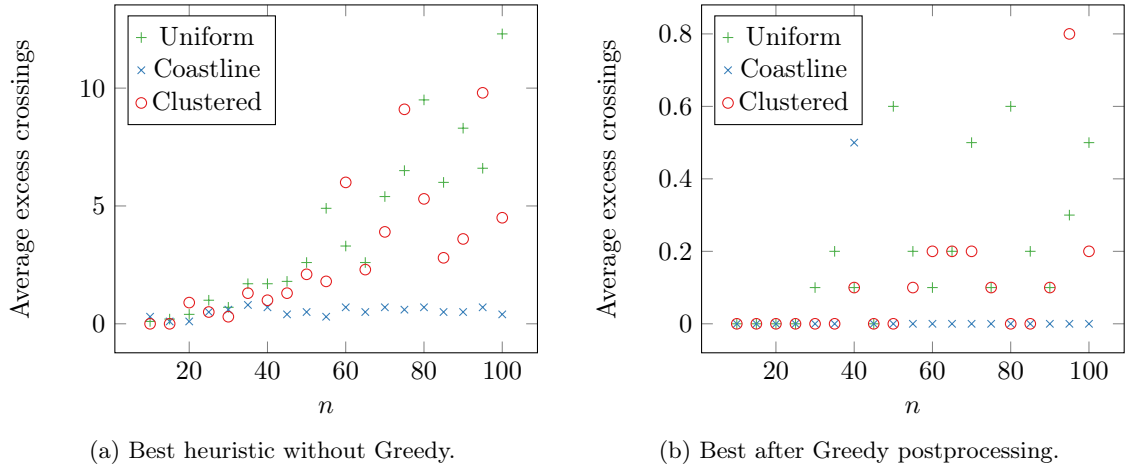
\begin{figure}
    \centering
    \begin{subfigure}{0.48\linewidth}
    \begin{tikzpicture}
    \begin{axis} [
        table/col sep=semicolon,
        width=\linewidth,
        legend pos=north west,
        xlabel={$n$},
        ylabel={Average excess crossings}
        ]
    \addplot [ only marks, mark=+, plot-uni ] table [x=n, y=uni-avg-heur-loss] {data-average.csv};
    \addlegendentry{Uniform}
    \addplot [ only marks, mark=x, plot-coast ] table [x=n, y=coast-avg-heur-loss] {data-average.csv};
    \addlegendentry{Coastline}
    \addplot [ only marks, mark=o, plot-cluster ] table [x=n, y=cluster-avg-heur-loss] {data-average.csv};
    \addlegendentry{Clustered}
    \end{axis}
    \end{tikzpicture}
    \caption{Best heuristic without Greedy.}
    \end{subfigure}
    \hfill
    \begin{subfigure}{.48\linewidth}
    \begin{tikzpicture}
        \begin{axis} [
            table/col sep=semicolon,
            width=\linewidth,
            legend pos=north west,
            xlabel={$n$},
            ylabel={Average excess crossings}
            ]
        \addplot [ only marks, mark=+, plot-uni ] table [x=n, y=uni-avg-greed-loss] {data-average.csv};
        \addlegendentry{Uniform}
        \addplot [ only marks, mark=x, plot-coast ] table [x=n, y=coast-avg-greed-loss] {data-average.csv};
        \addlegendentry{Coastline}
        \addplot [ only marks, mark=o, plot-cluster ] table [x=n, y=cluster-avg-greed-loss] {data-average.csv};
        \addlegendentry{Clustered}
        \end{axis}
        \end{tikzpicture}
        \caption{Best after Greedy postprocessing.}
    \end{subfigure}
    \caption{Average number of \sleader{} crossings made by the best heuristic minus the number of \sleader{} crossings in the optimal drawing, 
    averaged over 10 random instances per value of $n$.}
    \label{fig:excessxing}
\end{figure}

\paragraph{For the number of crossings, \poleaders are promising.}
Our heuristics require on average only about 73\% as many crossings when using \poleaders compared to \sleaders (55\% for Coastline instances);
the Lizard example in \cref{fig:examples:external} requires 11 \sleader crossings but only 2 \poleader crossings. 
We therefore propose that \poleaders deserve more attention from the phylogenetic community.

\paragraph{Algorithmic recommendations.}
Our results show that the ILP is a good choice for geophylogeny drawings with external labeling.
If no solver software is at hand or it is technically challenging to set up (for example when making an app that runs locally in a user's web browser), 
then the heuristics offer an effective and efficient alternative, especially Bottom-Up and Greedy.

For the Fish instance, for example, we found that the drawing with \sleaders and 17 crossings in \cref{fig:result:external:fish}
is a good alternative to the internal labeling used in the published drawing~\cite{realWorldExampleFish}.
However, for instances without a clear structure or with many crossings,
it might be better to use internal labeling.
Alternatively, the tree could be split like Tobler et~al.\ \cite{exampleSplit},
such that different subtrees are each shown with the map in separate drawings.

\begin{figure}[t]
  \centering
	\begin{subfigure}[t]{0.5 \linewidth}
		\centering
 		\includegraphics[height=6cm]{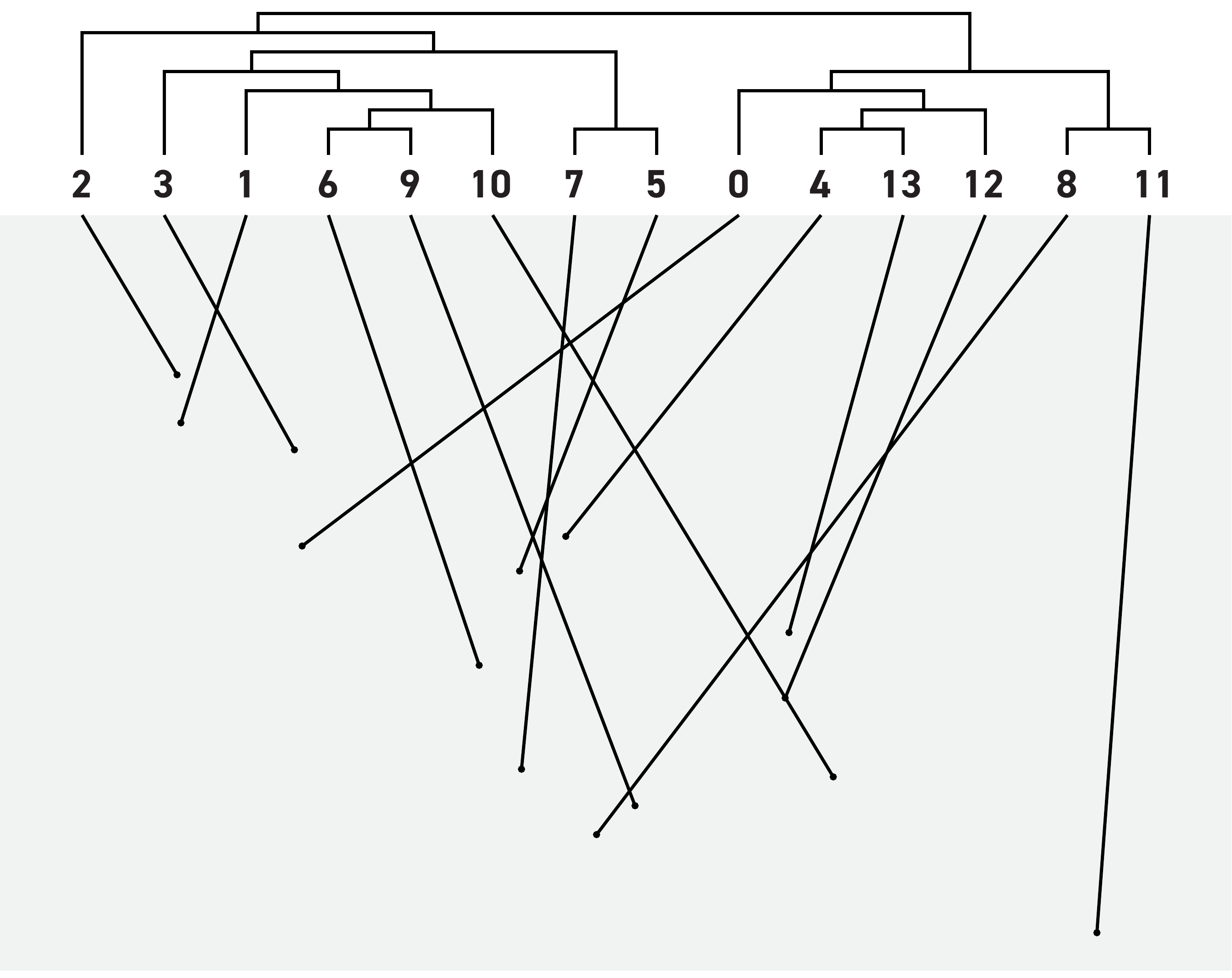}
		\caption{Drawing of \textbf{\sffamily Fish} with 17 crossings.}
		\label{fig:result:external:fish}
	\end{subfigure} 
	\hfill
    \begin{subfigure}[t]{0.4 \linewidth}
		\centering
 		\includegraphics[height=6cm]{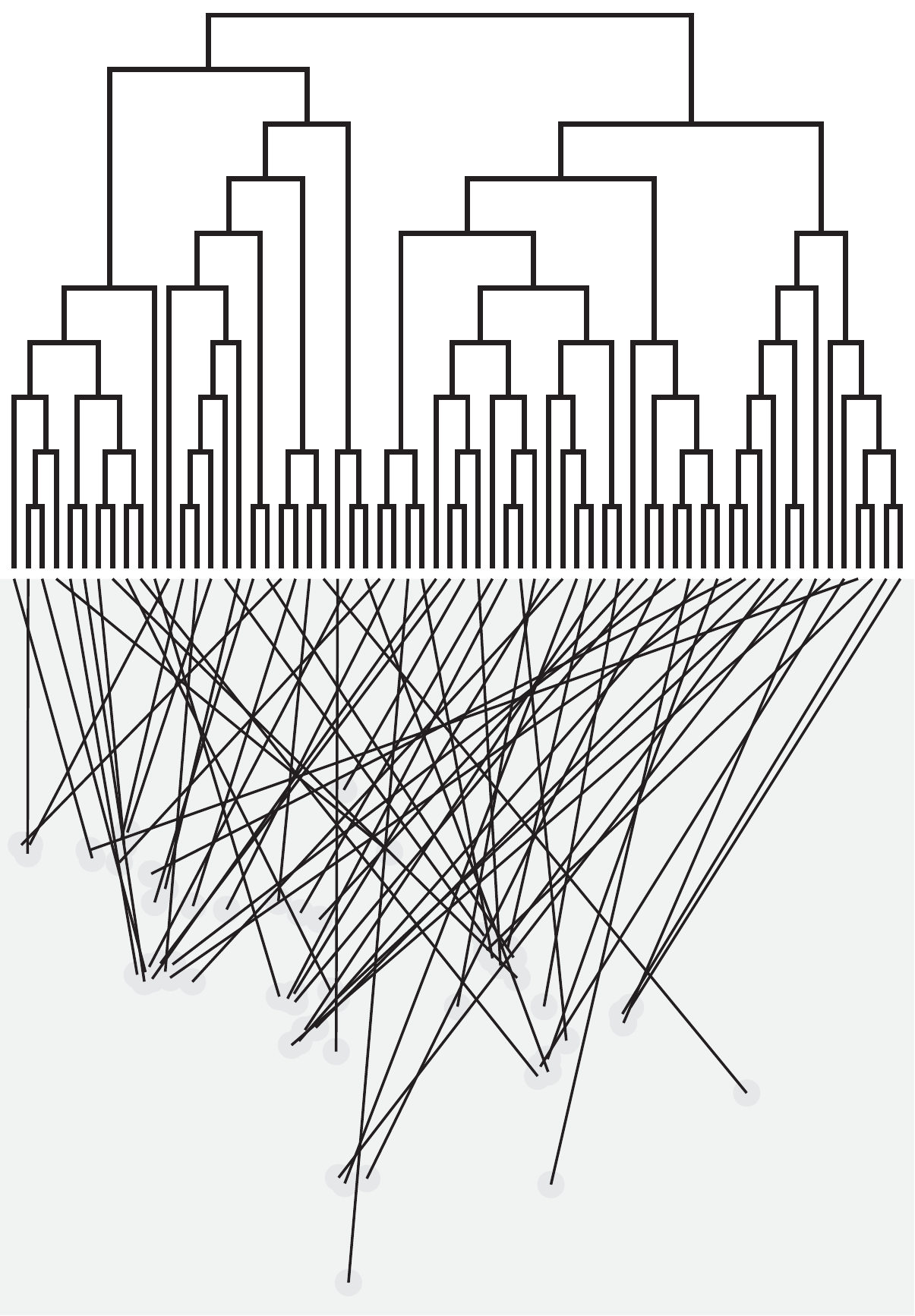}
		\caption{Drawing of \textbf{\sffamily Frogs} with 609 crossings.}
		\label{fig:result:external:frog}
	\end{subfigure}
  \caption{Crossing-optimal drawings of Fish and Frogs with \sleaders.}
  \label{fig:result:external}
\end{figure}

\section{Discussion and Open Problems} % -----------------------
\label{sec:conclusion}
In this paper, we have shown that drawings of geophylogenies can be approached theoretically and practically
as a problem of algorithmic map labeling.
We formally defined a drawing style for geophylogenies
that uses either internal labeling with text or colors,
or that uses external labeling with \spoleaders.
This allowed us to define optimization problems that can be tackled algorithmically.  
For drawings with internal labeling, we introduced a class of quality measures
that can be optimized efficiently and even interactively provided with user hints.
In practice, designers can thus try different quality measures, pick their favorite, 
and make further adjustments easily even for large instances.
For external labeling, minimizing the number of leader crossings is NP-hard in general.
Crossing free-instances on the hand can be found in polynomial time,
yet our algorithm still runs only in $\Oh(n^6)$ time.
Furthermore, for drawings with \sleaders, we showed that if the sites lie relatively close to a horizontal line
then in the best scenario an $\Oh(n \log n)$-time algorithm and otherwise an FPT algorithm can be used.
While we found similar results for drawings with \poleaders, 
it seems unlikely that geophylogenies arising in practice have the required properties.
Hence, we provide multiple algorithmic approaches to solve this problem 
and demonstrated experimentally that they perform well in practice.

Even though we have provided a solid base of results,
we feel the algorithmic study of geophylogeny drawings holds further promise by varying,
for example, the type of leader used, the quality measure, the composition of the drawing,
or the nature of the phylogeny and the map.
Several of these directions show parallels to the variations found in boundary labeling~problems. 
We finish this paper with several suggestions for future work.

One might consider \texttt{do}- and \texttt{pd}-leaders, 
which use a diagonal segment and can be aesthetically pleasing; see \cref{fig:do}. 
We expect that some of our results (such as the NP-hardness of crossing minimization 
and the effectiveness of the heuristics) should hold for these leader types.
The boundary labeling literature~\cite{BNN22} studies even further types, 
such as \texttt{opo} and Bézier, and these might be more challenging to adapt.

\begin{figure}[tbh]
  \centering
		\includegraphics[page=4]{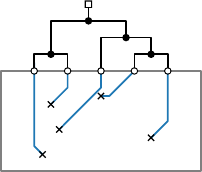}
  \caption{Drawing of a geophylogeny with \texttt{do}- and \texttt{pd}-leaders.}
  \label{fig:do}
\end{figure}

For external labeling we have only considered the total number of crossings.
If different colors are used for the leaders of different clades 
or if the drawing can be explored with an interactive tool,
one might want to minimize the number of crossings within each clade (or a particular clade).
Furthermore, one might optimize crossing angles or insist on a minimum distance between leaders and unrelated sites.
While we provided heuristics to minimize leader crossings,
the development of approximation algorithms, 
which exist for other labeling problems~\cite{LKY08,BKPS11}, could be of interest.

Our model of a geophylogeny drawing can be expanded as well.
One might allow the orientation of the map to be freely rotated, the extent of the map to be changed,
or the leaves to be placed non-equidistantly.
Optimizing over these additional freedoms poses new algorithmic challenges.
Straying further from our model, some drawings in the literature 
have a circular tree around the map~\cite{exampleCircular,exampleCircularTwo}; see \cref{fig:circular}.
This is similar to contour labeling in the context of map labeling~\cite{NNR17}. 
Also recall that \cref{fig:baseExample} has area features.
Our quality measures for internal labeling are easily adapted to handle this, 
but (as is the case with general boundary labeling~\cite{BKPS10}) 
area features provide additional algorithmic challenges for external labeling.  
The literature contains many drawings where multiple taxa correspond 
to the same feature on the map~\cite{exampleManyToOne} (see also again \cref{fig:circular})
and where we might want to look to many-to-one boundary labeling~\cite{LKY08,BCFHKNRS15}.
Furthermore, one can consider non-binary phylogenetic trees and phylogenetic networks.

\begin{figure}[tbh]
  \centering
  \includegraphics{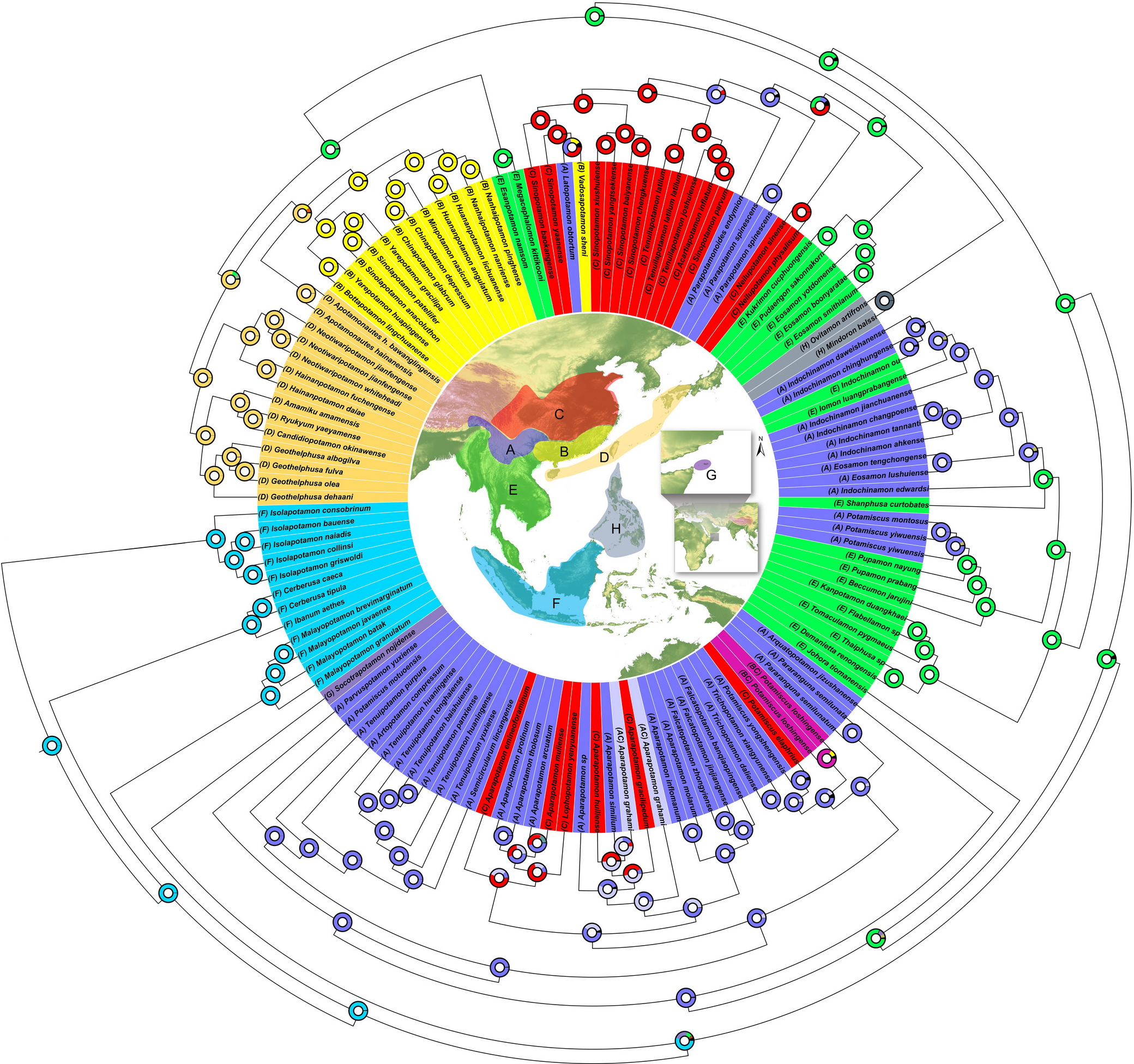}
  \caption{A drawing of a geophylogeny by Pan et~al.\ \cite{exampleCircular}, 
  where the tree is drawn circularly
  around the map and instead of a site per taxa only regions for clades are given.}
  \label{fig:circular}
\end{figure}

Lastly, we note that side-by-side drawings can also be used
for a phylogenetic tree together with a diagram other than a map:
Chen et~al.\ \cite{exampleDiagram} combine it with a scatter plot;
Gehring et~al.\ \cite{exampleThreeFigs} even combine three items (phylogenetic tree, haplotype network, and map).

\section*{Acknowledgements}
We thank the reviewers for their helpful comments and suggestions.

\clearpage
% ---- Bibliography ----
\pdfbookmark[1]{References}{References}
\bibliographystyle{abbrvurl}
\bibliography{sources}

\end{document}